\documentclass[twocolumn]{autart}    

\usepackage{graphicx}          
\usepackage{bm}
\usepackage{amsmath,amsfonts}
\usepackage{algorithmic}
\usepackage{algorithm}
\usepackage{array}
\usepackage[caption=false,font=normalsize,labelfont=sf,textfont=sf]{subfig}
\usepackage{textcomp}
\usepackage{stfloats}
\usepackage{url}
\usepackage{colortbl}
\usepackage{verbatim}
\usepackage{graphicx}
\usepackage{cite}
\usepackage{url}
\usepackage{hyperref}
\usepackage{appendix}
\usepackage{booktabs}
\usepackage{multirow}
\usepackage{hhline}
\usepackage{makecell}

\newcommand{\Cosh}{\mathrm{Cosh}}
\newcommand{\Tanh}{\mathrm{Tanh}}
\newcommand{\Sech}{\mathrm{Sech}}
\newcommand{\Ln}{\mathrm{Ln}}
\newcommand{\diag}{\mathrm{diag}}
\newtheorem{theorem}{Theorem}
\newtheorem{proof}{Proof}
\newtheorem{remark}{Remark}

\begin{document}
\begin{frontmatter}

\title{Payload trajectory tracking control for aerial transportation systems with cable length online optimization\thanksref{footnoteinfo}} 

\thanks[footnoteinfo]{
 This work was supported in part by National Natural Science Foundation of China under  Grant 623B2054, Grant 62273187 and Grant 62233011, in part by Natural Science Foundation of Tianjin under Grant 23JCQNJC01930, and in part by the Key Technologies R \& D Program of Tianjin under Grant 23YFZCSN00060.}
\thanks[Corresponding]{Corresponding author at: Institute of Robotics and Automatic Information System, Nankai University, Tianjin 300350, China. Tel.: +86 22 23505706; fax: +86 22 23500172.}
\author[Nankai,Shenzhen]{Hai Yu}\ead{yuhai@mail.nankai.edu.cn},    
\author[Nankai,Shenzhen]{Zhichao Yang}\ead{yangzc@mail.nankai.edu.cn},               
\author[Nankai,Shenzhen]{Wei He}\ead{howei@mail.nankai.edu.cn},  
\author[Nankai,Shenzhen]{Jianda Han}\ead{hanjianda@nankai.edu.cn},
\author[Nankai,Shenzhen]{Yongchun Fang}\ead{fangyc@nankai.edu.cn},
\author[Nankai,Shenzhen]{Xiao Liang\thanksref{Corresponding}}\ead{liangx@nankai.edu.cn}
\address[Nankai]{Institute of Robotics and Automatic Information System, College of Artificial Intelligence,
and Tianjin Key Laboratory of Intelligent Robotics, Nankai University, Tianjin 300350, China}  
\address[Shenzhen]{Institute of Intelligence Technology and Robotic Systems,
Shenzhen Research Institute of Nankai University, Shenzhen 518083, China}        


\begin{abstract}                          
Cable-suspended aerial transportation systems are employed extensively across various industries. 
The capability to flexibly adjust the relative position between the multirotor and the payload has spurred   growing interest in the system equipped with variable-length cable, promising broader application potential.  Compared to systems with fixed-length cables,  introducing the variable-length cable adds a new degree of freedom. However, it also results in increased nonlinearity and more complex dynamic coupling among the multirotor, the cable and the payload, posing significant challenges in control design. 
This paper introduces a backstepping control strategy  tailored for aerial transportation systems with variable-length cable, designed to precisely track the payload trajectory while dynamically adjusting cable length. 
Then, a cable length generator has been developed that achieves online optimization of the cable length while satisfying state constraints, thus balancing the multirotor's motion and cable length changes without the need for manual trajectory planning.
The asymptotic stability of the closed-loop system is guaranteed through Lyapunov techniques and the growth restriction condition. 
Finally, simulation results confirm the efficacy of the proposed method in managing trajectory tracking and cable length adjustments effectively.
\end{abstract}

\end{frontmatter}
 \renewcommand\baselinestretch{0.65}
\section{Introduction}
With the advancement of electronic technologies and control algorithms, unmanned aerial vehicles (UAVs) \cite{yang2024robust,kong2025robust,yang2023finite,meng2024physical,lv2022fixed} have become essential tools across various industries, including agriculture, logistics, and rescue operations, etc.
As one of the applications, using multirotor UAVs for suspended cargo transportation not only significantly enhances the efficiency of material delivery but also facilitates construction and rescue operations in environments inaccessible to ground vehicles. 
As a result, numerous studies have been conducted on cable-suspended aerial transportation systems in recent years.

Conventionally, the primary research objective is to achieve effective payload delivery across various scenarios through control and planning methods. In certain situations, inducing large swing angles of the payload is desirable, particularly during obstacle avoidance or while maneuvering through narrow spaces \cite{tang2018aggressive,yu2022aggressive,wang2024impact}. 
For instance, Tang \emph{et al.} \cite{tang2018aggressive} employ a downward-facing camera to estimate and capture the motion of the payload, facilitating rapid passage through slalom courses  by generating large payload swing angles.
Similarly, Yu \emph{et al.} \cite{yu2022aggressive}, integrate constraints on cable direction into their trajectory generator, creating an aggressive payload swing trajectory that enables efficient window crossing. 
Wang \emph{et al.}  \cite{wang2024impact}  have developed an impact-aware planning and control framework that combines agile flight with hybrid motion modes, allowing the aerial transportation system to navigate through narrow circular gates. 
Conversely,  for the purposes of transportation safety and payload protection, suppressing large payload swings is often necessary  \cite{xian2019online,lee2020antisway,yu2023reduced}. 
Xian \emph{et al.} \cite{xian2019online} design a trajectory planning strategy  that combines target positioning and antiswing components to suppress payload swings without iterative optimizations. 
Analogously, Lee \emph{et al.} \cite{lee2020antisway} introduce a dynamically feasible trajectory generator paired with an anti-swing tracking controller, which effectively dampens payload oscillations while allowing for transient aggressive motions. 
In addition, a control framework that tracks a virtual point along the cable, effectively handling constant external disturbances and mitigating payload oscillations, is implemented by Yu \emph{et al.} \cite{yu2023reduced}.
Furthermore, other studies have prioritized payload trajectory tracking over swing control \cite{cabecinhas2019trajectory,kong2024dynamic}. 
By using the full dynamics of the aerial transportation system, a backstepping control scheme is designed in \cite{cabecinhas2019trajectory} to track the desired position of the point-mass payload.
Besides, Kong \emph{et al.} \cite{kong2024dynamic} develop an adaptive backstepping control scheme for systems with unknown payload mass, setting pre-specified performance specifications for both payload position and cable direction.

However, most existing works focus on systems with fixed-length cables, which indicates that the payload's hoisting/lowering motion can only be achieved by changing the motion of the multirotor, making it difficult to complete certain tasks, such as navigating through narrow tunnels. Although some fixed-length cable systems can swing up the payload through narrow gaps using certain planning methods, it is not qualified for such tasks if the cave becomes longer.
Moreover,  fixed-length cable systems are not suitable for tasks that require the multirotor to maintain a certain distance from the payload, such as payload release and suspension missions involving human operators. Keeping the multirotor at a designated height while extending the cable length can enhance personnel safety. 
To address these limitations, some recent works have employed the variable-length cable for payload delivery, 
thereby significantly enhancing the system's flexibility. 
Specifically, the control objectives can be divided into two categories: one focusing on multirotor motion tracking \cite{liang2022unmanned,yu2023adaptive,yu2024visual,huang2023suppressing} and the other on payload motion tracking \cite{zeng2019geometric}. 
By mounting a motor beneath the multirotor, the cable length can be dynamically adjusted during flight, facilitating smooth navigation through narrow spaces and enabling sample collection tasks, as demonstrated by \cite{liang2022unmanned}. 
Furthermore, an adaptive control scheme is designed for the aerial transportation system to realize payload landing on mobile platforms by eliminating payload swing in \cite{yu2023adaptive}, and a visual-based servoing control method is proposed in \cite{yu2024visual} to improve the autonomy of payload transportation.
To expedite the convergence of payload swing angles, \cite{huang2023suppressing} introduces an enhanced coupling signal enabling the multirotor to track a desired point.

In summary, these efforts primarily focus on the motion tracking of the multirotor and the cable length, with particular emphasis on suppressing payload oscillations. While specific operations, such as the precise release of payloads at designated target points, necessitate direct manipulation of the payload. The traditional method \cite{yu2023adaptive} that indirectly controls payload position by suppressing its oscillations through adjustments to the multirotor and cable length exhibits limitations in task execution efficiency. 
To achieve direct payload position tracking, a geometric control scheme is designed in \cite{zeng2019geometric}, which  integrates the payload position and cable length into a single control unit. 
However, this approach requires predefined cable length trajectories, thus restricting the system's flexibility in task execution. 
For aerial transportation system with variable-length cable, the position of the payload is determined not only by the multirotor's motion but also by the length of the suspension cable. 
Actually, the same payload trajectory can be accomplished using different combinations of multirotor and cable length motions. 
Once the payload trajectory is established, it is necessary to consider the constraints related to multirotor and cable length motion  tailored to the specific requirements, to identify the optimal motion combination.
For instance, during the payload pickup and release phases, priority is given to extending the cable length to mitigate ground effect and ensure personnel safety. Conversely, during long-distance transportation phases, the focus shifts to multirotor movement to guarantee efficient transportation.
Currently, the coordination of multirotor motion with cable length variations to achieve precise payload trajectory tracking remains an open area of research.

To address these challenges, this paper presents a backstepping  controller accompanied by an online cable length generator for the aerial transportation system.
The developed control scheme  is capable of achieving  desired payload position tracking while simultaneously adjusting the cable length. The generator is specifically  designed to  balance the variations in cable length with multirotor movements. 
The main contributions of this paper can be summarized as follows:
\begin{enumerate}
\item  Although the variable cable length  increases the flexibility of aerial transportation systems, it also intensifies the nonlinearity and complex dynamics among the multirotor, cable and payload. By employing the backstepping procedure, a cascaded control scheme composed of four parts is designed to achieve payload trajectory tracking control, cable length and direction control, and multirotor attitude control. 
Without any linearization, the asymptotic stability of the closed-loop system is guaranteed through Lyapunov techniques and the growth restriction condition.

\item  After the payload trajectory is given, a cable length generator is designed to manage the motion of the multirotor and cable length based on state constraints, thereby eliminating the need for manually predefined cable trajectories. Notably, the developed control law is incorporated into the system model to determine the dynamic constraints of the generator, enabling real-time operation.

\end{enumerate}
The remainder of this paper is organized as follows: Section \ref{sec:dynamicsmodeling} presents the system modeling and problem statements. The controller design process is described in Section \ref{sec:control}, followed by the cable length generator in Section \ref{sec:generator}. Section \ref{sec:stability} provides the stability analysis, and Section \ref{sec:simulation_results} presents simulation results to demonstrate the performance of the proposed scheme. Finally, Section \ref{sec:conclusions} concludes the paper and outlines directions for future work.

\section{Problem Statement}\label{sec:dynamicsmodeling}
\subsection{Notations}
In this paper, vectors are distinguished by using bold letters. 
The special orthogonal group is defined as
$SO(3):=\{R\in\mathbb{R}^{3\times3}\mid R^\top R=I,\ \det R=1\}$, and its Lie algebra is
$\mathfrak{so}(3):=\{ S \in\mathbb{R}^{3\times3}\mid S^\top=-S\}$.
For $\bm a\in\mathbb{R}^3$, the hat map $\hat{\cdot}:\mathbb{R}^3\to\mathfrak{so}(3)$
is defined by $\hat{\bm a}\bm b=\bm a\times\bm b$ for all $\bm b\in\mathbb{R}^3$;
the vee map $(\cdot)^\vee:\mathfrak{so}(3)\to\mathbb{R}^3$ is its inverse.
\textcolor{blue}{The 2-sphere is given by $S^2:=\{\bm p\in\mathbb{R}^3\mid \bm p^\top\bm p=1\}$.} 
For a vector $\bm h \in \mathbb{R}^n$, define  $\Cosh(\bm h) = \left[\cosh(h_1),\ldots,\cosh(h_n) \right]^\top$,
 $\Tanh(\bm h) = \left[\tanh(h_1),\ldots,\tanh(h_n) \right]^\top$, 
 $\Sech(\bm h) = \left[\operatorname{sech}(h_1),\ldots,\operatorname{sech}(h_n) \right]^\top$, $\Ln(\bm h) = \left[\ln(h_1),\ldots,\ln(h_n) \right]^\top$, 
 and $\diag(\bm h) $ is the diagonal matrix with diagonal elements of $h_1,\ldots,h_n$.
 $\lambda_{\max}(\cdot) $ and $ \lambda_{\min}(\cdot)$ denote the maximum eigenvalue and minimum eigenvalue of a matrix, respectively. 
 For any vectors $\bm a,\bm b , \bm c \in \mathbb{R}^3$,  the following triple product identities are satisfied:
 $    \bm a \times (\bm b \times \bm c) = \bm b (\bm a^\top \bm c) - \bm c (\bm a^\top \bm b) $,
 $\bm a^\top(\bm b \times \bm c) = \bm b^\top (\bm c \times \bm a) = \bm c^\top (\bm a \times \bm b)
 $. For a unit vector $\bm p \in \mathbb{R}^3$ and an arbitrary vector $\bm a\in\mathbb{R}^3$, when  they are orthogonal to each other, i.e.,
 $\bm p^\top \bm a = 0$, one has
 $ -\hat{\bm p}^2\bm a = \bm a$.

\subsection{Dynamic Model}
\begin{figure}[!htp]
    \begin{center}
    \includegraphics[width=2.0in]{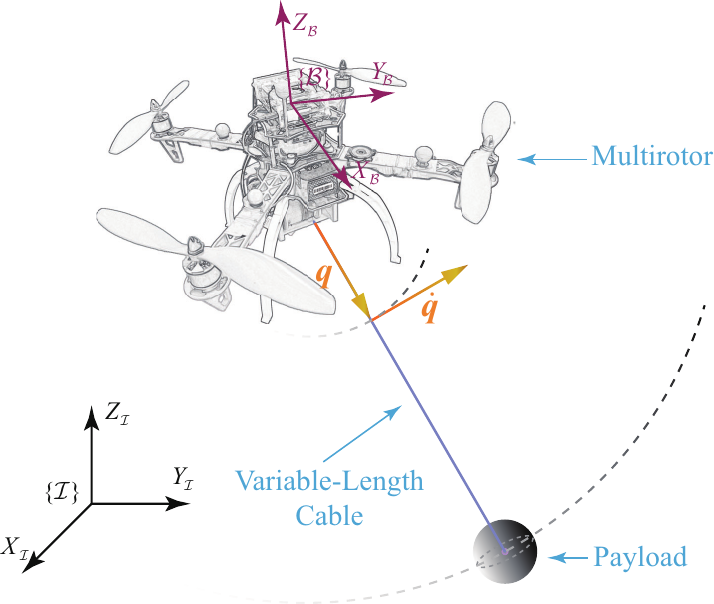}
    \caption{\small{Schematic of the aerial transportation system with variable-length cable.}}
    \label{fig:qload}
\end{center} 
\end{figure}

\begin{table}[t]
	\renewcommand{\arraystretch}{2.0}
	\caption{Symbols and Definitions of the System.}
	\centering
	\label{table:config}
	\resizebox{\columnwidth}{!}{
		\begin{tabular}{l l}
			\hline\hline \\[-4mm]
			\multicolumn{1}{l}{Symbol} & \multicolumn{1}{l}{Definition}  \\[0.2ex] \hline
            $m_Q,m_L \in \mathbb{R}$ & Mass of  multirotor and payload\\
            $g\in\mathbb{R}$ & Gravitational acceleration \\
            $L\in \mathbb{R}$ & Cable length \\
            $J \in\mathbb{R}^{3\times3}$ &Multirotor's moment of inertia  \\
            $R\in SO(3)$  & Rotation matrix from $\{\mathcal{B}\}$ to  $\{\mathcal{I}\}$ \\
            $\bm x_L,\bm v_L \in  \mathbb{R}^3$ &  Payload position and velocity \\
            $\bm x_Q,\bm v_Q \in  \mathbb{R}^3$ &  Multirotor position and velocity \\
            $\bm q \in S^2$ & Direction of the cable\\
            $\bm \omega \in \mathbb{R}^3 $ & Angular velocity of the cable\\
            $\bm\Omega \in \mathbb{R}^{3}$ & Angular velocity of multirotor in  frame $\{\mathcal{B}\}$\\
            $f_L\in \mathbb{R}$&  Payload hoisting/lowering force\\
            $f\in \mathbb{R}$ & Applied thrust generated by multirotor \\
            $\bm \tau \in \mathbb{R}^3$ &  Torque generated by multirotor\\
            \textcolor{blue}{$\bm e_3 \!=\! [0, 0, 1]^\top \in \mathbb{R}^3$ } & \textcolor{blue}{Unit vector along $z$-axis}\\
			\hline\hline
		\end{tabular}
	}
\end{table}
Fig. \ref{fig:qload} provides the schematic diagram of the aerial transportation system with variable-length cable, where ${\mathcal{I}}$ denotes the right-hand inertia frame and ${\mathcal{B}}$ represents the body-fixed frame.
The symbols and definitions utilized throughout the paper are summarized in Table \ref{table:config}. \textcolor{blue}{Unlike previous works, the cable length is modeled as a system state, rather than a constant, and can be actively extended or retracted via motor-driven mechanisms.  Therefore, the control inputs of the system comprise not only the thrust $f\in\mathbb{R}$ and torque $\bm \tau \in \mathbb{R}^3$ generated by the multirotor, but also the force $f_L\in\mathbb{R}$ for hoisting and lowering the payload produced by the cable-length regulation mechanism.}
\textcolor{blue}{
The system kinematics are given by \cite{sreenath2013geometric, tang2018aggressive, yu2023reduced}:
\begin{align}
    \label{Model:Qload1}
    \dot{\bm x}_{L}&=\bm v_{L}, \\
    \label{Model:Qload2}
    {\dot{\bm q}}&=\bm \omega\times \bm q, \\
    \label{Model:Qload3}
    \dot{R}&=R\hat{\bm \Omega}.
\end{align}
The relation between the multirotor and the payload is given by 
\begin{align}\label{Dynamics:relation}
    \bm x_Q = \bm x_L - L\bm q.
\end{align}
According to  Newton's second law, the dynamics of the payload and the multirotor are
\begin{align}
    \label{Model:Qload4}
m_L\dot{\bm v}_L =& f_L \bm q - m_L g \bm e_3,\\
\label{Dynamics:multirotor}
m_Q\dot{\bm v}_Q =& f R \bm e_3 - f_L \bm q - m_Q g \bm e_3. 
\end{align}
Taking the inner product of \eqref{Dynamics:multirotor} with $\bm q$, and replacing $\dot{\bm v}_Q$ with the second time derivative of \eqref{Dynamics:relation} yields
\begin{align}\label{Dynamics:relation2}
m_Q \left( \bm q^\top \dot{\bm v}_L - \ddot{L} +\dot{\bm q}^\top \dot{\bm q} \right) = \bm q^\top (f R \bm e_3 - f_L \bm q - m_Q g \bm e_3). 
\end{align}
Then, substituting \eqref{Model:Qload4} into \eqref{Dynamics:relation2} gives the cable-length variation dynamics as
\begin{align}\label{Model:Qload5}
\frac{m_Q+m_L}{m_L} f_L - m_Q \ddot{L} = \bm q^\top f R \bm e_3 - m_Q L (\dot{\bm q}^\top \dot{\bm q}). 
\end{align}  
The multirotor is regarded as a point mass $m_Q$ located at the position vector $\bm r : = \bm x_Q - \bm x_L = - L \bm q$ with respect to the payload center $O$. Thus, the relative  angular momentum $\bm H$ about the payload center $O$ is expressed as:
\begin{align}\label{Dynamics:angular_momentum}
    \bm H  := \bm r \times m_Q (\bm v_Q - \bm v_L)  =  m_Q L^2\bm \omega.
\end{align}
By using the angular momentum theorem with respect to the accelerating reference point $O$ \cite{hinrichsen2005mathematical}, one has
\begin{align}\label{Dynamics:angular_momentum_theorem}
    \dot{\bm H} & =  \bm \tau_O -\bm r \times m_Q \bm a_O,
\end{align}
where $\bm \tau_O : = \bm r \times \left(fR\bm e_3 - f_L \bm q - m_Q g \bm e_3\right) $ is the total torque acting on the multirotor about $O$, and $\bm a_O:=\dot{\bm v}_L$ is the acceleration of the payload center.
Substituting \eqref{Dynamics:angular_momentum_theorem} into the time derivative of \eqref{Dynamics:angular_momentum} yields the cable direction dynamics
\begin{align}\label{Model:Qload6}
    m_Q L \dot{\bm \omega} = -\bm q \times f R \bm e_3 - 2m_Q \dot{L} \bm \omega.
\end{align}
The rotational dynamics of the multirotor \cite{lee2010geometric}, described by the Newton-Euler equation \cite{spong2020robot}, are
\begin{align}
    \label{Model:Qload7}
     J\dot{\bm \Omega}&=\bm \tau-\bm \Omega\times J\bm \Omega.
\end{align}
Collecting \eqref{Model:Qload1}--\eqref{Model:Qload3}, \eqref{Model:Qload4}, \eqref{Model:Qload5}, \eqref{Model:Qload6}, \eqref{Model:Qload7}, the complete dynamic model of the aerial transportation system with variable-length cable is established.} 
Since this paper focuses on payload trajectory tracking, the model is constructed based on the payload's position and velocity. 
By performing a variable substitution using \eqref{Dynamics:relation}, the system model \eqref{Model:Qload1}--\eqref{Model:Qload3}, \eqref{Model:Qload4}, \eqref{Model:Qload5}, \eqref{Model:Qload6}, \eqref{Model:Qload7} can be reformulated in terms of the multirotor's position and velocity.

\subsection{Control Objective and Error Dynamics}
For the convenience of subsequent control scheme design and analysis,  the state errors of the system are first defined as follows.
The payload position and velocity errors are  given by:
$$
    \bm e_{x}:= \bm x_L - \bm x_{Ld},
    \bm e_{v} :=\bm v_L - \dot{\bm x}_{Ld},
$$
where $\bm x_{Ld}\in\mathbb{R}^3$ represents the desired payload trajectory.  The cable direction  and  angular velocity errors are defined as:
$$
    \bm e_q := \bm q_d \times \bm  q,  \bm e_{\omega} := \bm \omega + \hat{\bm q}^2\bm \omega_d,
    $$
where $\bm q_d\in S^2$ and $\bm \omega_d := \bm q_d \times \dot{\bm q}_d\in\mathbb{R}^3$ denote the desired cable direction and angular velocity, respectively.
The cable length error is defined as:
$$
    e_L := L-L_d,
    $$
where $L_d\in\mathbb{R}$ is the desired cable length.
The multirotor attitude and angular velocity errors are defined as:
$$
    \bm e_R := \frac{1}{2}\left(R_d^\top R - R^\top R_d\right)^\vee, \bm e_{\Omega} := \bm \Omega - R^\top R_d \bm \Omega_d,
    $$
where $R_d\in SO(3)$  is the desired rotation matrix and $\bm \Omega_d\in\mathbb{R}^3$ is the desired multirotor angular velocity.

The control objective is to design a feedback control law that drives the payload to its desired trajectory,
while simultaneously adjusting the cable length and direction to achieve their respective desired states. This can be mathematically described as follows:
\begin{align}
     &\bm e_x \rightarrow \bm 0_{3\times 1}, \bm e_v \rightarrow \bm 0_{3\times 1},
     \bm e_q \rightarrow \bm 0_{3\times 1}, \bm e_{\omega} \rightarrow \bm 0_{3\times 1},\nonumber\\
      &e_L \rightarrow 0, \bm e_R \rightarrow \bm 0_{3\times 1}, \bm e_{\Omega} \rightarrow \bm 0_{3\times 1}.\nonumber
\end{align}
Subsequently, by substituting the system errors into the system dynamics model \eqref{Model:Qload1}--\eqref{Model:Qload3}, \eqref{Model:Qload4}, \eqref{Model:Qload5}, \eqref{Model:Qload6}, \eqref{Model:Qload7}, the open-loop error dynamics of the system can be rearranged as:
\begin{align}
    \label{Model:Oev}
    m_L\dot{\bm e}_v = &f_L \bm q- m_L\ddot{\bm x}_{Ld} - m_Lg\bm e_3,\\
    \label{Model:OeL}
    m_Q\ddot{e}_L = & - \bm q^\top f R\bm e_3 +\frac{m_Q+m_L}{ m_L} f_L  + m_QL\left(\dot{\bm q}^\top\dot{\bm q}\right)\nonumber\\
    &- m_Q\ddot{L}_d,
     \\
     \label{Model:Oeomega}
     m_QL \dot{\bm e}_\omega = &-\bm q \times fR \bm e_{3}-2m_Q\dot{L}\bm \omega + m_QL(\bm q^\top \bm \omega_d)\dot{\bm q} \nonumber\\
    & + m_QL (\bm \omega_d^\top \dot{\bm q}
    +\bm q^\top \dot{\bm \omega}_d)\bm q - m_QL \dot{\bm \omega}_d,
    \\
    \label{Model:OeOOmega}
     J{{\dot{\bm e}}_{\Omega }} = &\bm \tau  -\bm\Omega\times J \bm\Omega+J\left( \hat{\bm\Omega}{{R}^{\top }}{{R}_{d}}{{\bm \Omega }_{d}}-{{R}^{\top }}{{R}_{d}}{{{\dot{\bm \Omega }}}_{d}} \right).
\end{align}

\section{Controller Design}\label{sec:control}
In this section, the design process of the payload trajectory tracking controller is presented through the backstepping procedure. The control scheme is developed from four aspects: payload position control, cable length control, cable direction control, and multirotor attitude control. The block diagram of the proposed algorithm is depicted in Fig. \ref{fig:control_diagram}. The pink blocks represent the desired payload trajectory and cable length trajectory generator modules. The green blocks denote the control modules. The purple blocks illustrate the calculations conducted between these modules. Finally, the blue block represents the aerial transportation system.
\subsection{Payload Position Control}

\begin{figure*}[!thp]
    \begin{center}
    \includegraphics[width=6.6in]{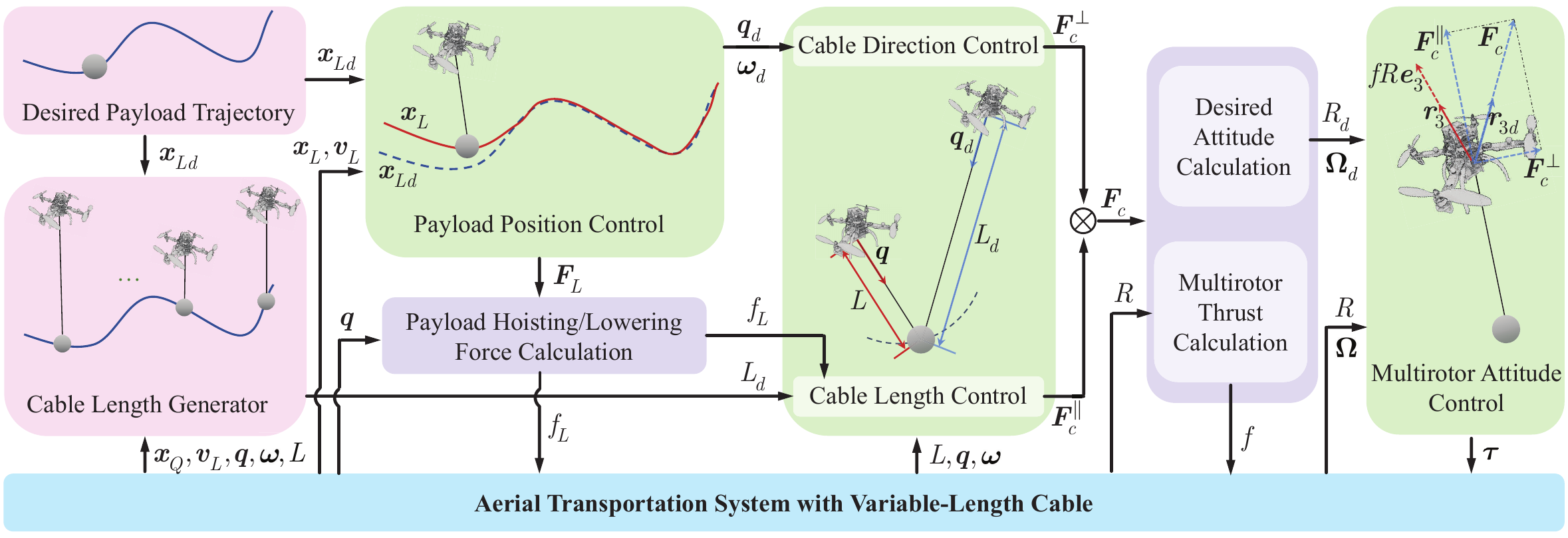}
    \caption{\small{Block diagram of the designed payload trajectory tracking algorithm.}}
    \label{fig:control_diagram}
\end{center} 
\end{figure*}
The payload position control scheme is designed based on the error dynamics \eqref{Model:Oev}.
According to the hierarchical control approach \cite{kendoul2009nonlinear,zhao2014nonlinear},
let $f_L \bm q: =\bm F_L + \bm \Delta_L $, where $\bm F_{L}\in \mathbb{R}^3$  is the virtual control input to be constructed,  and $\bm \Delta_L \in\mathbb{R}^3 $ is  the coupling term between the  virtual input $\bm F_{L}\in \mathbb{R}^3$ and
the cable direction error $\bm e_q$. The decomposition of $f_L \bm q$ is given by
\textcolor{blue}{\[
f_L\bm q
=\underbrace{\frac{f_L}{\bm q_d^{\!\top}\bm q}\,\bm q_d}_{\bm F_L}
\;+\;
\underbrace{\frac{f_L}{\bm q_d^{\!\top}\bm q}\Big[(\bm q_d^{\!\top}\bm q)\,\bm q-\bm q_d\Big]}_{\boldsymbol{\Delta L}}.
\]
As seen from the above decomposition, $ \bm F_L $  and $ \bm \Delta_L $ can be obtained by multiplying the same scalar $\tfrac{f_L}{\bm q_d^{\top}\bm q}$ with the vectors $\bm q_d$ and $(\bm q_d^{\top}\bm q)\bm q-\bm q_d$, respectively. Consequently, the geometric relationship between $\bm q_d$ and $(\bm q_d^{\top}\bm q)\bm q-\bm q_d$ in Fig. \ref{fig:DeltaL_FL} is directly analogous to that between $\bm F_L$ and $\bm \Delta_L$.
\begin{figure}[!thp]
    \begin{center}
    \includegraphics[width=2.4in]{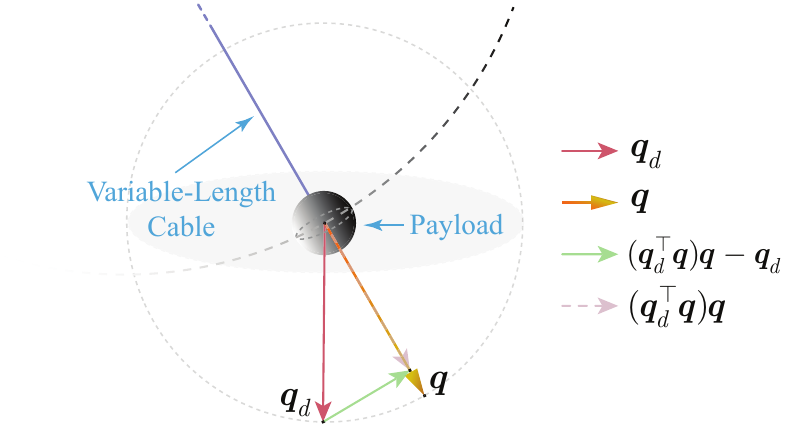}
    \caption{Geometric illustration of $\bm q$, $\bm q_d$, and  $(\bm q_d^{\top}\bm q) \bm q-\bm q_d$.}
    \label{fig:DeltaL_FL}
\end{center} 
\end{figure}
Using the vector triple-product identity, $\bm \Delta_L$ can be further expressed as: 
$
\bm \Delta_L 
= -\,\frac{f_L}{\bm q_d^\top\bm q}\,\bm q\times(\bm q_d\times\bm q)
= -\,\frac{f_L}{\bm q_d^\top\bm q}\,\bm q\times\bm e_q. 
$
Taking norms and recalling $\|\bm F_L\|=\tfrac{|f_L|}{|\bm q_d^\top\bm q|}\|\bm q_d\| =\tfrac{|f_L|}{|\bm q_d^\top\bm q|} $ yields
\begin{align}\label{Stability:NormdeltaL}
\|\bm \Delta_L\|
=\|\bm F_L\|\,\|\bm q\times\bm e_q\|\le \|\bm F_L\|\,\|\bm e_q\|.
\end{align}}
The payload position control scheme $\bm F_L$ is designed as:
\begin{align}\label{Control:Fq}
    \bm F_L =& - K_p \Tanh(\bm e_x)  - K_d \Tanh(\bm e_v)  \nonumber\\
    &+  m_L \ddot{\bm x}_{Ld} + m_L g\bm e_3,
\end{align}
where diagonal matrices $K_p = \mathrm{diag}\left(\left[k_{px},k_{py},k_{pz}\right]\right)$ and $K_d=\mathrm{diag}([k_{dx}, k_{dy}, k_{dz}]) \in \mathbb{R}_+^{3\times3}$
are positive-definite gain matrices.   Subsequently, the desired cable direction $\bm q_d $  can be obtained by
\begin{align}
    \bm q_d =& - \frac{\bm F_L}{\|\bm F_L\|} .
\end{align}
Thus, according to the decomposition of $f_Lq$, one has
\begin{align}
  \bm F_L = \frac{f_L}{\bm q_d^{\!\top}\bm q}\,\bm q_d = - \frac{f_L \|\bm F_L\| }{\bm F_L^{\!\top}\bm q}\, \frac{\bm F_L}{\|\bm F_L\|} =  \frac{f_L  }{\bm F_L^{\!\top}\bm q}\,\bm F_L,\nonumber
\end{align}
which indicates that the payload hoisting/lowering force $f_L$ can be derived by
\begin{align}
     f_L =& \bm F_L^\top \bm q .
\end{align}

\subsection{Cable Length Control}
From the error dynamic model \eqref{Model:OeL} and \eqref{Model:Oeomega},  it is evident that the term $fR \bm e_3$ impacts both the length and direction dynamics of the cable.
Specifically, the component of  $fR \bm e_3$ along $\bm q$ governs the cable length, whereas its component perpendicular to $\bm q$ dictates the direction of the cable.
Define the virtual control input of $fR \bm e_3$ as $\bm F_c$.
To decouple the cable length and direction control, $\bm F_c$ is divided into two parts that satisfy the following equation:
\begin{align}\label{Control:Fc}
    \bm F_c : = \bm  F_c^\parallel +  \bm F_c^\perp,
\end{align}
where $\bm  F_c^\parallel$ denotes the projection of $\bm F_c$ along $\bm q$, and
$\bm F_c^\perp$
is perpendicular to $\bm F_c^\parallel$.

Thereafter,  the control design for the cable length will be addressed. 
According to  the above analysis and the error dynamics \eqref{Model:OeL}, the cable length control law $f_c := -\bm q^\top \bm F_c^\parallel$ can be constructed as:
\begin{align}\label{Control:fc}
    f_c  = & -k_{pl} e_L - k_{dl} \dot{e}_L-\frac{m_Q+m_L}{m_L}   f_L \nonumber\\
    &- m_Q L (\dot{\bm q}^\top\dot{\bm q}) + m_Q  \ddot{L}_d - \frac{k_a e_L}{\iota^2 -e_L^2},
\end{align} 
where  $k_{pl}$, $k_{dl}$ and $k_a\in\mathbb{R}_+$ are  positive control gains, $\iota\in\mathbb{R}_+$ is a positive constant that specifies the admissible upper bound on the cable-length error.

\subsection{Cable Direction Control}
Based on the analysis of \eqref{Control:Fc} and the error dynamics \eqref{Model:Oeomega}, the cable direction control law $\bm F_c^\perp$ can be designed as:
\begin{align}\label{Control:Fcperp}
    \bm F_c^\perp   = & m_QL\cdot\bm q \times\left( - k_q \bm e_q -k_\omega \bm e_\omega - (\bm q^\top \bm \omega_d )\dot{\bm q}\right.\nonumber\\
    &\left.
     - \hat{\bm q}^2  \dot{\bm \omega}_d +2 \frac{\dot{L}}{L}\bm \omega
     -\frac{k_b \bm e_q}{\varrho^2 -\Psi_q }\right),
\end{align}
where $\Psi_q :=1-\bm q^\top \bm q_d$ is the configuration error for the $S^2$ manifold, $k_q,k_\omega $ and  $k_b\in\mathbb{R}_+$ are positive control gains, $ \varrho \in (0,1)$ is a positive constant that specifies the admissible upper bound on the cable-direction misalignment. Thus, according to \eqref{Control:Fc}, by collecting \eqref{Control:fc} and \eqref{Control:Fcperp}, the applied multirotor control  thrust input can be obtained as $f = \bm F_c^\top R\bm e_3$.

\subsection{Multirotor Attitude Control}
The desired multirotor attitude can be defined as:
\begin{align}
     R_d : = \left[\bm r_{1d}; \bm r_{2d}; \bm r_{3d}\right],
\end{align}
where $\bm r_{3d} \in S^2$ is calculated by $\bm r_{3d} = \frac{\bm F_c}{\|\bm F_c\|}$, $\bm r_{2d} = \frac{\bm r_{3d} \times \bm r_{1a}}{\|\bm r_{3d} \times \bm r_{1a}\|}$, and $\bm r_{1d} = \bm r_{2d} \times \bm r_{3d}$.
$\bm r_{1a}$ is an arbitrary vector that is not parallel to $\bm r_{3d}$. The desired angular velocity can be calculated as $\bm \Omega_d = ( R_d^\top \dot{R}_d)^\vee$.
Following the approach similar to that in \cite{sreenath2013geometric}, the multirotor attitude control is implemented as: 
\begin{align}\label{Control:tau}
    \bm \tau =&-\frac{1}{\epsilon^{2}}k_{R}\bm e_{R}-\frac{1}{\epsilon}k_{\Omega}\bm e_{\Omega}+ \bm \Omega\times J_{Q} \bm \Omega  \nonumber\\
&-J_{Q}(\hat{\bm \Omega}R^{T}R_{d}\bm \Omega_{d}-R^{T}R_{d}\dot{\bm \Omega}_{d}),
\end{align}
where $k_{R},k_{\Omega}\in\mathbb{R}_+$ are positive control gains, and $0<\epsilon<1$. Through the
singular perturbation theory \cite{sreenath2013geometric}, the  parameter $\epsilon$
is introduced to enable rapid exponential convergence of $\bm e_R$ and $\bm e_\Omega$.

\textcolor{blue}{From the above subsections, one can observe that the desired cable direction $\bm q_d$, cable angular velocity $\bm \omega_d$, rotation matrix $R_d$, and multirotor angular velocity $\bm \Omega_d$ depend on the virtual control inputs $\bm F_L$ and $\bm F_c$ and their time derivatives. Fortunately, these quantities can be computed from the closed-loop system dynamics via systematic derivations, without requiring acceleration or higher-order derivatives of the states. Consequently, the proposed method is readily deployable on real systems, with no need for numerical differentiation of the states.}

\section{Cable Length Generator}\label{sec:generator}
With the introduction of the variable-length cable, the relative distance between the payload and the multirotor can be freely adjusted. 
Once the payload trajectory has been established, the motion of the multirotor and changes in cable length can be tailored according to the specific characteristics of the task.
For example, during the  payload pickup and release  process, the movement of the multirotor should be minimized to ensure precise handling, with the payload managed by extending and retracting the cable. In contrast, during long-distance transportation, the emphasis shifts to the multirotor's movement, while minimizing variations in the cable length. Building upon this concept, this section designs the cable length generator.
According to the aforementioned backstepping control schemes, the cable length trajectory is set to be fourth-order continuous to ensure the smoothness of the control input. A generalized state vector $\bm \xi$ is introduced to describe the  nonlinear optimization problem, i.e., $\bm \xi : = (\bm x_Q, \bm q, L, \bm v_Q, \bm \omega, \dot{L}, \bm L_t) \in \mathbb{R}^3\times S^2 \times \mathbb{R}\times \mathbb{R}^3 \times  \mathbb{R}^3 \times  \mathbb{R} \times  \mathbb{R}^5$, where $\bm L_t :=\left[L_d,\dot{L}_d,\ddot{L}_d,L_d^{(3)},L_d^{(4)}\right]^\top\in\mathbb{R}^5$ is the generalized desired cable length trajectory  including $L_d$ as well as its first to fourth time derivatives. $\bm X_d : = ( \bm x_{Ld}, \dot{\bm x}_{Ld}, \ddot{\bm x}_{Ld},  {\bm x}_{Ld}^{(3)}, {\bm x}_{Ld}^{(4)}) \in \mathbb{R}^{15}$ is the generalized desired trajectory vector of the payload. To obtain a reasonable cable length trajectory, the general form of the cable length generator is designed as follows:
\begin{align}
    \label{gen:costfunction}
    \min_{L_d^{(5)}}\,\,&\int_{t}^{t + T}  \ell (\bm \xi(\delta), {L}_d^{(5)}(\delta)) \,\mathrm{d}\delta  \\ 
    \mathrm{s.t.}\,\,
    \label{gen:constraint1}
    &\dot{\bm \xi} = f_\xi(\bm \xi,  \bm X_d, L_d^{(5)} ),  \\
    \label{gen:constraint2}
    & \underline{\bm L}_t \leq \bm L_t\leq \overline{\bm L}_t,  \\
    \label{gen:constraint3}
    & \underline{ L}_d^{(5)} \leq  L_d^{(5)}\leq \overline{ L}_d^{(5)},
\end{align}
where \eqref{gen:costfunction} is the cost function, \eqref{gen:constraint1} is the generalized dynamic constraint, \eqref{gen:constraint2} is the constraint of the generalized desired cable length, and \eqref{gen:constraint3} is the constraint of the output. 
\textcolor{blue}{The cable length generator is model-based and controller-aware. Under the proposed backstepping control schemes, the cable-length trajectory must be $C^4$ continuous to ensure smooth control inputs. Accordingly,  the fifth derivative of the desired cable length $L_d^{(5)}$, is chosen as the optimization output. 
Since the cable-length generator is devoted to coordinating the multirotor motion with cable-length variations, a generalized dynamical system \eqref{gen:constraint1} is adopted as the dynamic constraint of the generator. Its state vector comprises the multirotor motion; the cable length and its direction; and the desired cable length along with its first-fourth derivatives.
According to the relationship between the multirotor and the payload \eqref{Model:Qload4}, the dynamic constraint \eqref{gen:constraint1} can be obtained by substituting the payload position control scheme \eqref{Control:Fq}, along with the cable length and direction control laws \eqref{Control:fc}, \eqref{Control:Fcperp} into the dynamic model \eqref{Model:Qload4}, \eqref{Model:Qload5} and \eqref{Model:Qload6}.} 
\eqref{gen:constraint2} is the constraint of the generalized desired cable length, where $\underline{\bm L}_t$, $\overline{\bm L}_t \in \mathbb{R}^5$ are the lower and upper bounds, respectively.
Constraint \eqref{gen:constraint3} governs the output, specifying the lower bound $\underline{L}_d^{(5)}$ and the upper bound $\overline{ L}_d^{(5)}$. 
Tailored to different task requirements, the cost function $\ell$ is designed to enhance coordination between the multirotor's motion and cable length variations, integrating this with the system's state.

\section{Stability Analysis}\label{sec:stability}
The stability of the closed-loop system is analyzed in this section.  
By neglecting the coupling term $\bm \Delta_L$,  the  stability of the  payload position, cable length and direction
are first proven individually. Subsequently, based on the theory of cascade systems, the stability  of the overall system 
is ensured  by demonstrating that the coupling term satisfies the growth restriction condition 
\cite{kendoul2009nonlinear,zhao2014nonlinear}.

\begin{theorem}\label{Theorem:position}
    The designed payload position control law \eqref{Control:Fq} guarantees the convergence of the payload position and velocity errors to zero
    asymptotically, i.e.,
    \begin{align}
        \lim_{t \rightarrow \infty}\bm e_x =  \bm 0_{3\times1},   \lim_{t \rightarrow \infty}\bm e_v =  \bm 0_{3\times1}.\nonumber
    \end{align}
\end{theorem}
\begin{proof}
\textcolor{blue}{To prove \emph{Theorem \ref{Theorem:position}}, the following positive-definite function is chosen as the Lyapunov function candidate:
\begin{align}
    \label{Control:V_1}
    V_1=  \frac{1}{m_L} \bm k_p^\top \Ln \left[\Cosh(\bm e_x)\right] + \frac{1}{2}\bm e_v^\top\bm e_v  ,
\end{align}
where $\bm k_p=\left[k_{px},k_{py},k_{pz}\right]^\top\in \mathbb{R}_+^{3}$ is the position gain vector satisfying $K_p = \diag(\bm k_p)$.
Taking the time derivative of \eqref{Control:V_1}, neglecting the coupling term $\bm \Delta_L$, and substituting \eqref{Model:Oev} and \eqref{Control:Fq} yields:
\begin{align}
    \label{Stability:dotV_1}
   \dot{V}_1 
    =  &\frac{1}{m_L} \bm e_v^\top\left(K_p  \Tanh(\bm e_x) + \bm F_L - m_L\ddot{\bm x}_{Ld} -m_L g\bm e_3  \right) \nonumber\\
   = & -\bm e_v^\top K_d \Tanh(\bm e_v) \nonumber\\
   \leq & 0 .
\end{align}}
\textcolor{blue}{Therefore, from \eqref{Control:Fq}, \eqref{Control:V_1} and \eqref{Stability:dotV_1}, one knows $V_1(t) $ is nonincreasing and bounded, which implies that $V_1(t) \leq V_1(0) < \infty$ for all $t \geq 0$. It follows  that $\bm e_x $ and $\bm e_v$ are both bounded, i.e.,
$
    \bm e_x, \bm e_v \in \mathcal{L}_\infty.
$ }

\textcolor{blue}{To further analysis the convergence of the payload position and velocity errors, define $ \varepsilon_0:=\bm e_v^\top K_d\Tanh(\bm e_v)$, $\bm \varepsilon_1 : = - \frac{1}{m_L}K_p \Tanh(\bm e_x)$,  $\bm \varepsilon_2 : = -\frac{1}{m_L} K_d \Tanh(\bm e_v)$, which are bounded, i.e., $ \varepsilon_0, \bm \varepsilon_1, \bm \varepsilon_2 \in \mathcal{L}_\infty$. 
Substituting \eqref{Control:Fq} into the error dynamics \eqref{Model:Oev}, the closed-loop dynamics can be expressed as:
$
    \dot{\bm e}_v =\bm \varepsilon_1 + \bm \varepsilon_2,
$
which implies $\dot{\bm e}_v \in \mathcal{L}_\infty$. 
Integrating \eqref{Stability:dotV_1} with respective to time yields:
$
\int_{0}^{\infty}\varepsilon_0(t)\,dt
=V_1(0)-\lim_{t\to\infty}V_1(t)<\infty,
$
so $\varepsilon_0 \in \mathcal{L}_1$. Since  $\bm e_v, \dot{\bm e}_v \in\mathcal L_\infty$, 
it follows that $\dot \varepsilon_0=\dot{\bm e}_v^\top K_d\Tanh(\bm e_v)
+\bm e_v^\top K_d \diag\!\left(\Sech^2(\bm e_v)\right)\dot{\bm e}_v\in\mathcal L_\infty$, hence $\varepsilon_0$ is uniformly continuous. According to Barbalat's lemma \cite{khalil2002nonlinear},  $\varepsilon_0$ is uniformly continuous and $\varepsilon_0 \in \mathcal{L}_1$, one can conclude that $\varepsilon_0\to 0$ as $t\to\infty$.
Thus, it follows that
\begin{align}\label{Stability:ev0}
  \lim_{t \to \infty} \bm e_v =\bm 0.
\end{align}
Since
$
\dot{\bm \varepsilon}_1 = - \frac{1}{m_L}K_p \diag\big(\Sech^2(\bm e_x)\big)\bm e_v\in\mathcal{L}_\infty,
$
one knows $\bm \varepsilon_1$ is uniformly continuous.
Based on \eqref{Stability:ev0}, one can achieve $\lim_{t \to \infty}\bm \varepsilon_2  = \bm 0$.
In summary, one has $\lim_{t \to \infty} \bm e_v = \bm 0$, $\bm \varepsilon_1$ is uniformly continuous, and $\lim_{t \to \infty}\bm \varepsilon_2 = \bm 0$. According to the extended Barbalat's Lemma \cite{Behal2010}, one can conclude that $\lim_{t \to \infty}\bm \varepsilon_1 = \bm 0 $. So we can derive that
\begin{align}\label{Stability:ex0}
\lim_{t \to \infty} \bm e_x = \bm 0.
\end{align}
Thus, by collecting the result in \eqref{Stability:ev0} and \eqref{Stability:ex0}, the proof is completed.}
\end{proof}

\begin{theorem}\label{Theorem:cable}
    Consider the closed-loop system under the control laws \eqref{Control:fc} and \eqref{Control:Fcperp}. 
    \textcolor{blue}{Let the barrier Lyapunov functions be defined on the safe sets $\mathcal S_L:=\{|e_L|<\iota\}$ and $\mathcal S_q:=\{\Psi_q<\varrho^2\}$, where $\iota\in \mathbb R_{+}$ and $\varrho\in(0,1)$ are positive constants. 
    If the \emph{initial errors} satisfy $|e_L(0)|<\iota$ and $\Psi_q(0)<\varrho^2$,
then there exist positive control parameters
$k_{pl},k_{dl},k_a,k_q,k_\omega,k_b\in\mathbb R_{+}$ such that: }

    \textcolor{blue}{1) $|e_L(t)|<\iota$ and $\Psi_q(t)<\varrho^2$ for all $t\ge0$;}

    2) the cable length and direction tracking errors converge exponentially to zero, i.e.,
    \begin{align}
        \lim_{t \rightarrow \infty}e_L\! = \!  0,  \lim_{t \rightarrow \infty}\dot{e}_L\!  =  \! 0, \lim_{t \rightarrow \infty}\bm e_q\!  =\! \bm 0_{3\times1},
        \lim_{t \rightarrow \infty}\bm e_\omega \! =\! \bm 0_{3\times1}.\nonumber
    \end{align}
\end{theorem}

\begin{proof}
Consider the singularly perturbed model with time-scale parameter $\varepsilon$.
The multirotor attitude is exponentially stable and time-scale separated.
In the reduced slow model obtained by letting $\varepsilon\to 0$ \cite{sreenath2013geometric}, the fast states satisfy $R=R_d$, so the thrust vector equals the commanded virtual force:
$
fRe_3 = fR_de_3 = \bm F_c= \bm  F_c^\parallel +  \bm F_c^\perp.
$
Hence, on the slow manifold, the cable length and direction dynamics evolve with $R\equiv R_d$ and the realized thrust matches the commanded virtual force.

\textcolor{blue}{By selecting
$0<\beta_L < \min \left\{ \sqrt{\frac{k_{pl}}{m_Q}}, \frac{4k_{pl}k_{dl}}{4m_Qk_{pl} + k_{dl}^2} \right\}$, the barrier Lyapunov function candidate is chosen as
\begin{align}\label{Control:V_2}
    V_2 \!= \! \frac{k_{pl}}{2m_Q}  e_L^2 \!+\!\beta_L e_L \dot{e}_L\! +\! \frac{1}{2}  \dot{e}_L^2 \!+\!
 \frac{k_a}{2m_Q} \ln \frac{\iota^2}{\iota^2\! -\! e_L^2},
\end{align}
which is well-defined on $\mathcal S_L:=\{ |e_L|<\iota\}$.
If the initial condition satisfies $|e_L(0)|<\iota$, it indicates that $V_2(0)<\mathcal{L}_\infty$.
Thus, for all $t>0$, one has $|e_L(t)|<\iota$,
and the subsequent derivations are carried out on the interior of $\mathcal S_L$.
Let $\bm z_L:= \left[e_L,\dot{e}_L\right]^\top$,  it is straightforward to show that
\begin{align}
    \label{Stability:V_2}
    V_2
    = &\frac{1}{2} \bm z_L ^\top N_L \bm z_L +
    \frac{k_a}{2m_Q} \ln \frac{\iota^2}{\iota^2 - e_L^2}\nonumber\\
\leq & \frac{1}{2} \bm z_L ^\top N_L \bm z_L + \frac{k_a}{2m_Q} \frac{e_L^2}{\iota^2 - e_L^2}.
\end{align}
where $N_L = \begin{bmatrix}\frac{k_{pl}}{m_Q} & \beta_L\\ \beta_L & 1\end{bmatrix}$.
Taking the time derivative of $V_2$ on the interior of $\mathcal S_L$ and substituting dynamic \eqref{Model:OeL}  and controller \eqref{Control:fc} into the result yields
\begin{align}  \label{Stability:dotV_2}
    \dot{V}_2
        = &  \frac{k_{pl}}{m_Q}  e_L\dot{e}_L + \beta_L  \dot{e}_L^2 + (\beta_L e_L + \dot{e}_L)\cdot
     \left(\frac{m_Q+m_L}{m_Q m_L } f_L\right. \nonumber\\
     & \left. -   \ddot{L}_d - \frac{1}{m_Q}\bm q^\top \bm F_c^\parallel + L\left(\dot{\bm q}^\top\dot{\bm q}\right)\right)
     + \frac{k_a}{m_Q}\frac{e_L\dot{e}_L}{\iota^2 -e_L^2}\nonumber\\
     =& -\frac{\beta_L k_{pl}}{m_Q} e_L^2 - \frac{\beta_L k_{dl}}{m_Q} e_L \dot{e}_L -(\frac{k_{dl}}{m_Q}-\beta_L)\dot{e}_L^2\nonumber\\
     &- \frac{k_a \beta_L e_L^2}{m_Q(\iota^2 -e_L^2)}\nonumber\\
     =&- \bm z_L ^\top W_L \bm z_L - \frac{k_a \beta_L e_L^2}{m_Q(\iota^2 -e_L^2)},
\end{align}
where $W_L = \begin{bmatrix}\frac{\beta_L k_{pl}}{m_Q} & \frac{\beta_L k_{dl}}{2m_Q}\\ 
    \frac{\beta_L k_{dl}}{2m_Q} & \frac{k_{dl}}{m_Q}-\beta_L\end{bmatrix}$.} 
\textcolor{blue}{The term $- \frac{k_a \beta_L e_L^2}{m_Q(\iota^2 -e_L^2)}$  blows up as $|e_L|\to\iota$, thereby enforcing forward invariance of $\mathcal{S}_L$. Thus, by choosing 
$L_d$ and $\iota$ appropriately, the cable length can be guaranteed to remain strictly positive for all time, without imposing the assumption that it is always positive \cite{yu2023adaptive}.}
Selecting $\alpha_L := \min\left\{2\frac{\lambda_{\min}(W_L)}{\lambda_{\max}(N_L)} , 2 \beta_L\right\}$,
one can obtain that
\begin{align}\label{Stability:V_2_0}
    \dot{V}_2 \leq -\alpha_L V_2.
\end{align}
Therefore, the cable length control law \eqref{Control:fc} can guarantee the exponential convergence of the cable length error to zero.

\textcolor{blue}{Furthermore,  selecting the following scalar function related to the cable direction and angular velocity as the barrier Lyapunov function candidate:
\begin{align}\label{Control:V_3}
    V_3 = &k_q\Psi_q  + \frac{1}{2} \bm e_{\omega}^\top \bm e_{\omega}
    + \beta_q \bm e_q^\top\bm e_\omega +k_b\ln\frac{\varrho^2}{\varrho^2- \Psi_q},
\end{align}
where $0<\beta_q < \min \left\{ \sqrt{k_q}, \frac{4k_qk_\omega}{4k_q + (k_\omega+C_\omega)^2} \right\}$, and $C_\omega = \sup\left(\left\|(2I_{3\times3}-\bm q \bm q^\top)\bm \omega_d\right\|\right)$ is a positive constant. 
$V_3$ is well-defined on the safe set $\mathcal S_q:=\{\Psi_q<\varrho^2\}$. The initial condition $\Psi_q(0)<\varrho^2$ indicates that $V_3(0)<\mathcal{L}_\infty$. Thus, for all $t>0$, one has $\Psi_q(t)<\varrho^2$, and the subsequent derivations are carried out on the interior of $\mathcal S_q$.
By the definitions of $\bm e_q$ and $\Psi_q$, they satisfy the relationship $\bm e_q^\top \bm e_q = \Psi_q(2-\Psi_q)$, i.e.,
$
\Psi_q=\frac{\|\bm e_q\|^2}{\,2-\Psi_q\,}.
$
Thus, under the condition $ \Psi_q < \varrho^2$ with $\varrho\in(0,1)$, it follows that $2-\Psi_q \in (2-\varrho^2,2)$. 
Then, the following inequality can be obtained:
\begin{align}
    \frac{1}{2} \|\bm e_q\|^2 \leq \Psi_q < \frac{1}{2-\varrho^2} \|\bm e_q\|^2.
\end{align}
Let $\bm z_q : = \left[\|\bm e_q\|,\|\bm e_\omega\|\right]^\top$,  it is straightforward to show that
\begin{align}
    \label{Stability:V_3_0}
    V_3\geq &
    \frac{1}{2} \bm z_q ^\top N_{q1} \bm z_q +k_b\ln\frac{\varrho^2}{\varrho^2- \Psi_q}, \\
    \label{Stability:V_3}
     V_3 
     \leq & \frac{1}{2} \bm z_q ^\top N_{q2} \bm z_q +k_b \frac{\Psi_q}{\varrho^2- \Psi_q}\nonumber\\
        < & \frac{1}{2} \bm z_q ^\top N_{q2} \bm z_q + \frac{k_b}{2-\varrho^2} \frac{\|\bm e_q\|^2}{\varrho^2- \Psi_q},
\end{align}
where $N_{q1} = \begin{bmatrix}k_q & -\beta_q\\ -\beta_q & 1\end{bmatrix}$, $N_{q2} = \begin{bmatrix}
    \frac{2k_q}{2-\varrho^2} & \beta_q\\ \beta_q & 1
\end{bmatrix}$.
Taking the time derivative of  $V_3$ on the interior of $\mathcal S_q$ and inserting  dynamic \eqref{Model:Oeomega}  and controller \eqref{Control:Fcperp} into the result, one achieves
\begin{align}  \label{Stability:dotV_3}
    \dot{V}_3
        =& \left(\bm e_{\omega}\! + \!\beta_q \bm e_q\right)\left(\!-\frac{1}{m_QL}\bm q \times \bm F_c^\perp \!-\!2\frac{\dot{L}}{L}\bm \omega \! +\!
    (\bm q^\top \bm \omega_d)\dot{\bm q}\right. \nonumber\\
    &+  \left(\bm \omega_d^\top \dot{\bm q}
    +\bm q^\top \dot{\bm \omega}_d\right)\bm q -  \dot{\bm \omega}_d \Bigg)\!\!+ \beta_q \dot{\bm e}_q^\top \bm e_\omega
     + k_q \bm e_q^\top \bm e_\omega \nonumber\\
     &+ \frac{k_b\bm e_q^\top \bm e_\omega}{\varrho^2 -\Psi_q}\nonumber\\
    =& -k_q\beta_q\bm e_q^\top\bm e_q - k_\omega \beta_q\bm e_q^\top \bm e_\omega -  k_\omega\bm e_\omega^\top\bm e_\omega +
    \beta_q \dot{\bm e}_q^\top \bm e_\omega\nonumber\\
    &- \frac{k_b \beta_q \bm e_q^\top \bm e_q }{ \varrho^2 -\Psi_q}\nonumber\\
    \leq& -k_q\beta_q \|\bm e_q\|^2 + k_\omega \beta_q\|\bm e_q\| \|\bm e_\omega\| -  k_\omega \|\bm e_\omega\|^2 \nonumber\\
    &+\beta_q \|\bm e_\omega\|^2 + \beta_q C_\omega \|\bm e_q\| \|\bm e_\omega\| - \frac{k_b \beta_q \|\bm e_q\|^2}{ \varrho^2 -\Psi_q}\nonumber \\
    = & -\bm z_q ^\top W_q \bm z_q - \frac{k_b \beta_q \|\bm e_q\|^2}{ \varrho^2 -\Psi_q},
\end{align}
where $W_q = \begin{bmatrix} \beta_q k_q&  -\frac{1}{2}\beta_q( k_\omega + C_\omega)\\
    - \frac{1}{2}\beta_q( k_\omega +  C_\omega) & k_\omega -\beta_q\end{bmatrix}$.
Please refer to Appendix \ref{App:A} for the calculations related to $ \dot{\bm e}_q^\top \bm e_\omega $.} 
\textcolor{blue}{The term $-\dfrac{k_b \beta_q \|\bm e_q\|^2}{\varrho^2-\Psi_q}$ blows up as $\Psi_q \to \varrho^2$, thereby rendering the set $\mathcal S_q := \{\Psi_q<\varrho^2\}$ forward invariant. Consequently, unlike prior works that bound $\Psi_q$ to exclude the antipodal configuration \cite{sreenath2013geometric,lee2017geometric,goodman2022geometric}, the stability analysis requires no such assumption in this paper.
Moreover, when $\bm q_d$ is in the lower hemisphere and $\varrho$ is chosen appropriately, the cable direction is guaranteed to remain within the lower hemisphere, thereby obviating the assumption that the payload must always lie beneath the multirotor \cite{liang2022unmanned,yu2023adaptive,yu2024visual}.}
Let $\alpha_q : = \min\left\{2\frac{\lambda_{\min}(W_q)}{\lambda_{\max}(N_{q2})} ,  \beta_q(2-\varrho^2)\right\}$,
one can achieve that
\begin{align}\label{Stability:dotV_3_0}
    \dot{V}_3 \leq -\alpha_q V_3.
\end{align}
Therefore, 
the zero equilibrium of the
the cable direction error $\bm e_q, \bm e_\omega$ is exponentially stable. 

By collecting the results in \eqref{Stability:V_2_0} and \eqref{Stability:dotV_3_0}, \emph{Theorem \ref{Theorem:cable}} is proved.
\end{proof}
\textcolor{blue}{
\begin{remark}
    Note that $\beta_L$ and $\beta_q$ are chosen so that $V_2$ and $V_3$ are positive definite with respect to the error states $\bm z_L$ and $\bm z_q$, respectively, and that their time derivatives $\dot V_2$ and $\dot V_3$ are negative definite. Hence, the cable-length and direction tracking errors can converge exponentially to zero. To this end, the matrix inequalities $N_L \succ 0$, $W_L \succ 0$, $N_{q1} \succ 0$, and $W_q \succ 0$ must hold. By applying Sylvester's criterion \cite{horn2012matrix}, explicit conditions on $\beta_L$ and $\beta_q$ can be obtained.
\end{remark}}

\begin{theorem} \label{Theorem:overall}
    The proposed control law \eqref{Control:Fq}, \eqref{Control:fc}, \eqref{Control:Fcperp} and \eqref{Control:tau} can drive the payload
     to the desired trajectory, and adjust the cable length and direction to the desired ones,  implying the following result:
    \begin{align}
        \lim_{t \rightarrow \infty}\left[\bm e_x^\top, \bm e_v^\top, e_L, \dot{e}_L,\bm e_q^\top, \bm e_\omega^\top, \bm e_R^\top, \bm e_\Omega^\top\right]^\top
        =  \bm 0_{20\times1}.\nonumber
    \end{align}
\end{theorem}
\begin{proof}
    The stability analysis of the overall closed-loop system with the consideration of the coupling term $\bm \Delta_L$
    is taken into account here.  Let the generalized payload position and velocity error vector
    as $\bm z_x : = \left[\bm e_x^\top,\bm e_v^\top\right]^\top$. 
   Furthermore, utilizing the fact  $|\tanh(\cdot)|<|\cdot| $, the  virtual control input \eqref{Control:Fq} satisfies that
\begin{align}
    \|\bm F_L\|=&\|\!- \!K_p \Tanh(\bm e_x)  \!- \!K_d \Tanh(\bm e_v) \!  +\!  m_L \ddot{\bm x}_d \!+\! m_L g\bm e_3\|\nonumber\\
\leq& \max(\lambda_{\max}(K_p),\lambda_{\max}(K_d))\left(\|\bm e_x\| +\|\bm e_v\|\right) + \Gamma\nonumber\\
\leq& \sqrt{2}\max(\lambda_{\max}(K_p),\lambda_{\max}(K_d))\|\bm z_x\|+ \Gamma\nonumber\\
\leq& \left[\frac{\Gamma}{\sqrt{2}\max(\lambda_{\max}(K_p),\lambda_{\max}(K_d))}+\|\bm z_x\|\right]\cdot\nonumber\\
&\sqrt{2}\max(\lambda_{\max}(K_p),\lambda_{\max}(K_d)),\nonumber
\end{align}
where $\Gamma = \sup\left( \|  m_L \ddot{\bm x}_d + m_L g\bm e_3\|\right)$.
Then, by selecting $k_f=2\sqrt{2}\max(\lambda_{\max}(K_p),\lambda_{\max}(K_d))$,
$c_f = \frac{\Gamma}{\sqrt{2}\max(\lambda_{\max}(K_p),\lambda_{\max}(K_d))}$, one can obtain that   $\bm F_L$
 satisfies  the following property:
\begin{align}\label{Stability:theorem3}
    \left\|\bm F_L(\bm z_x)\right\| \leq \begin{cases}
        k_f\|\bm z_x\|, & \text { for }\|\bm z_x\| \geq c_f \\
        k_f c_f, & \text { for }\|\bm z_x\|<c_f\end{cases}.
\end{align}
Subsequently, substituting the result of \eqref{Stability:theorem3} into \eqref{Stability:NormdeltaL},
the coupling term $\bm \Delta_L $ can be further deduced as:
\begin{align}
\|\bm \Delta_L\|
& \leq\|\bm F_L\|\|\bm e_q\|\nonumber\\
&\leq k_f \|\bm z_x \|\|\bm e_q\|, \text { for }\|\bm z_x\| \geq c_f,\nonumber
\end{align}
which means that the coupling term $\bm \Delta_L$ satisfies the growth restriction condition
\cite{kendoul2009nonlinear,zhao2014nonlinear}. Finally,  combining the conclusions from \emph{Theorem \ref{Theorem:position}} and \emph{Theorem \ref{Theorem:cable}},
 along with the stability analysis results on the attitude of unmanned aerial vehicles as stated in reference \cite{sreenath2013geometric},
  the proof of \emph{Theorem \ref{Theorem:overall}}  can be completed.
\end{proof}

\textcolor{blue}{
\begin{remark}
    Notably, in the presence of model uncertainties or external disturbances, a rigorous theoretical robustness guarantee for the proposed control scheme cannot yet be established. Accordingly, robustness validation is assessed by simulation in Section \ref{sec:sim4}. Future work will strengthen robustness in two directions: (i) building on robust controllers for fixed-length cable systems, such as robust integral of the sign of the error (RISE) control \cite{cai2024robust} and sliding-mode control \cite{liang2024robust}, and extending them to variable-length settings by explicitly modeling cable length dynamics and cable length dependent couplings; and (ii) adapting robust designs developed for other cascaded systems, e.g., geometric control on SE(3) \cite{lee2013nonlinear} and nested-saturation controllers \cite{naldi2016robust} for multirotors, to the  aerial transportation system with variable-length cable.
\end{remark}}
\begin{figure*}[!tp]
    \centering
        \hspace*{-10mm} 
    \subfloat[Test1: $k_1 = 0.1$, $k_2 = 100$.]{
        \includegraphics[height=2.45in]{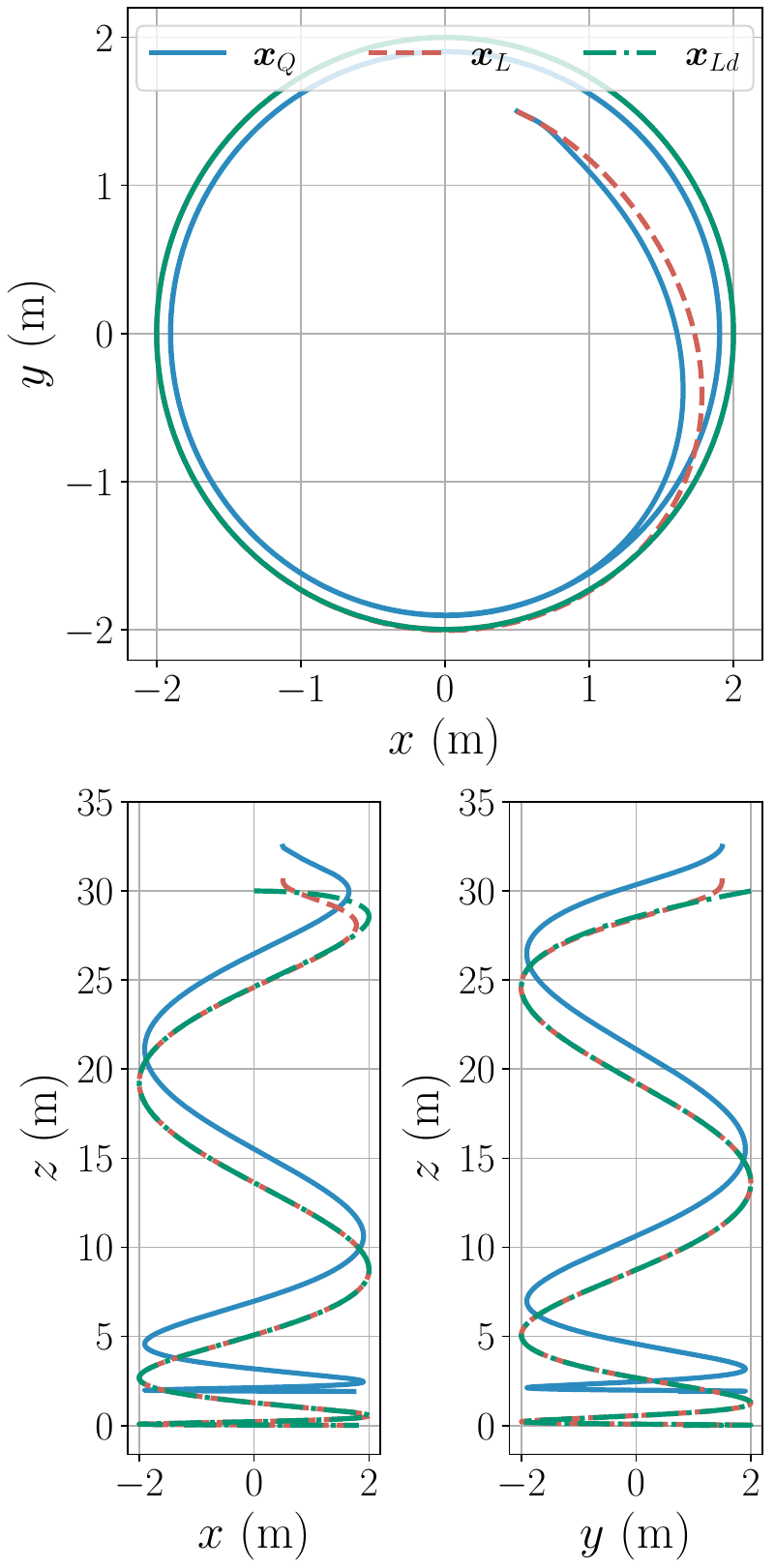}
        \includegraphics[height=2.45in]{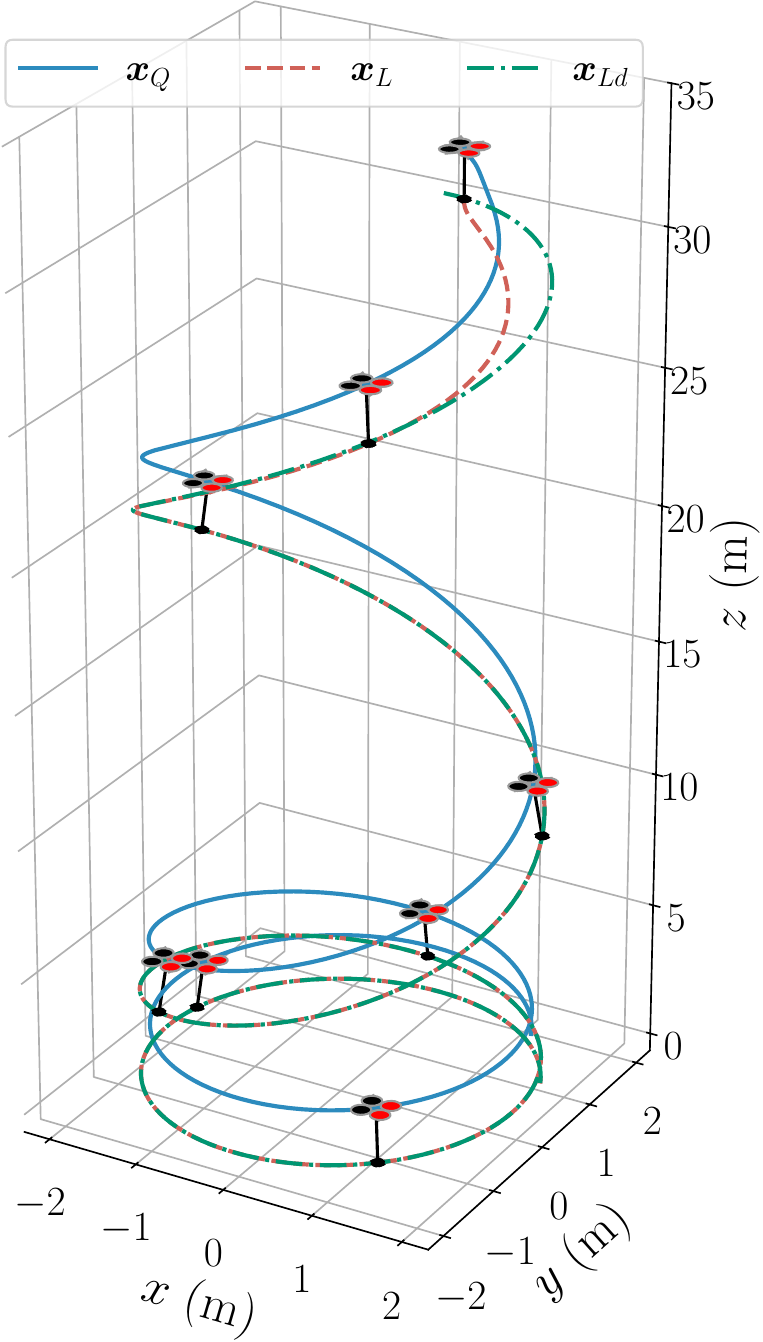}
        \includegraphics[height=2.45in]{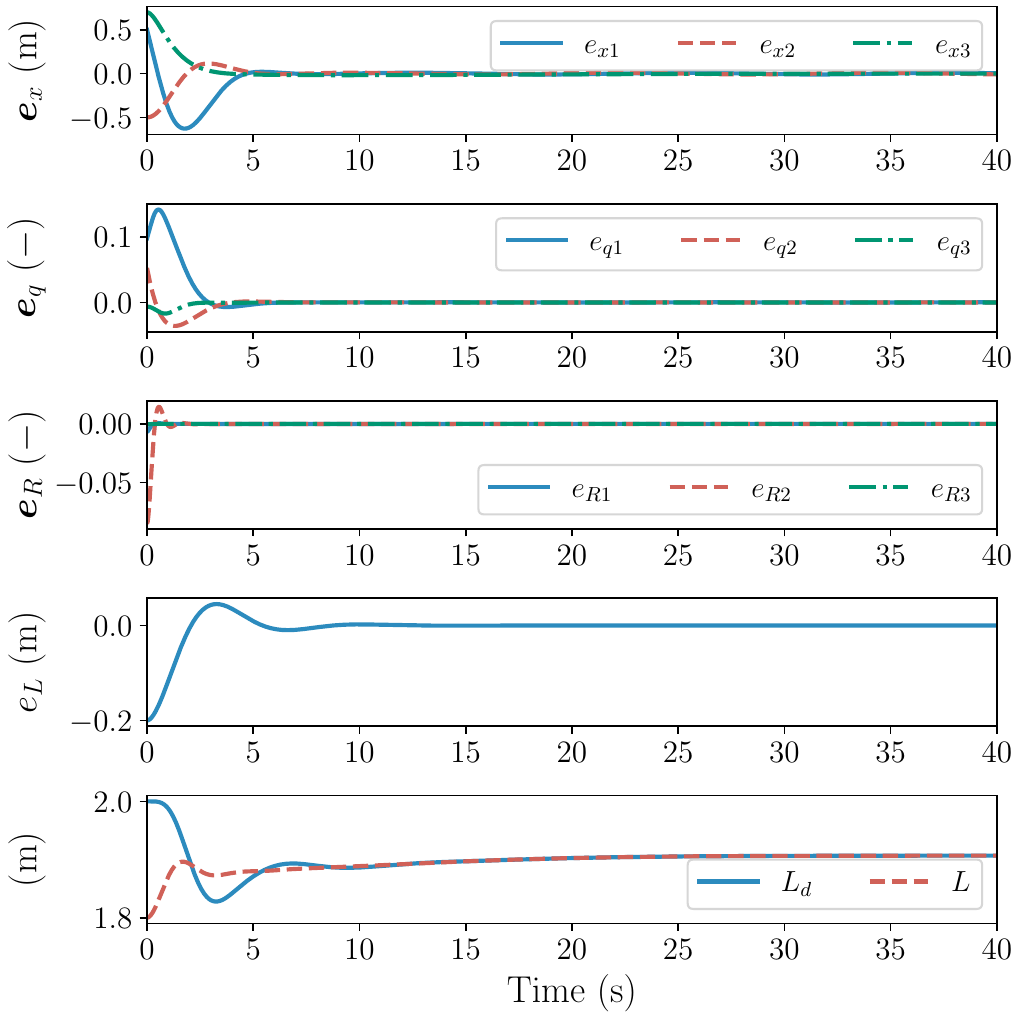}
        \includegraphics[height=2.45in]{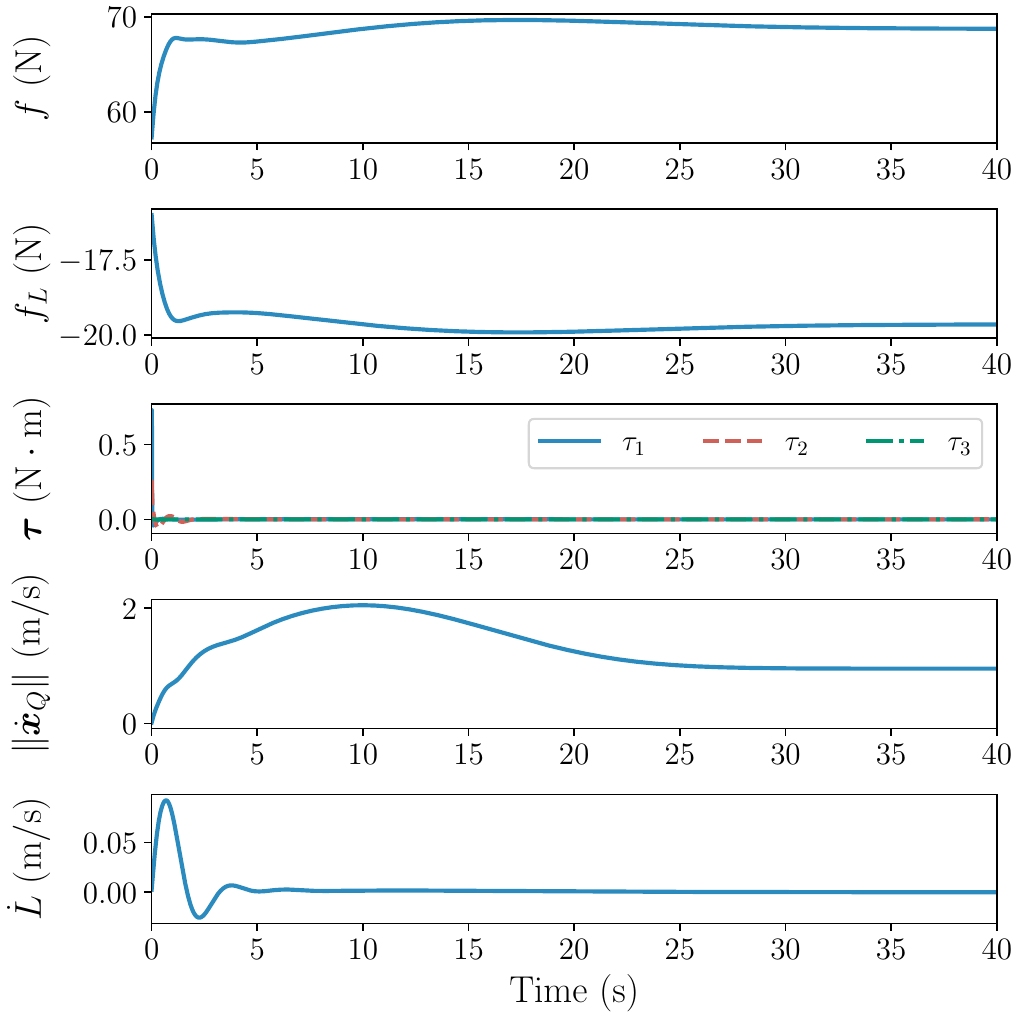}
    }\\
        \hspace*{-10mm} 
    \subfloat[Test2: $k_1 = 100$, $k_2 = 100$.]{
        \includegraphics[height=2.45in]{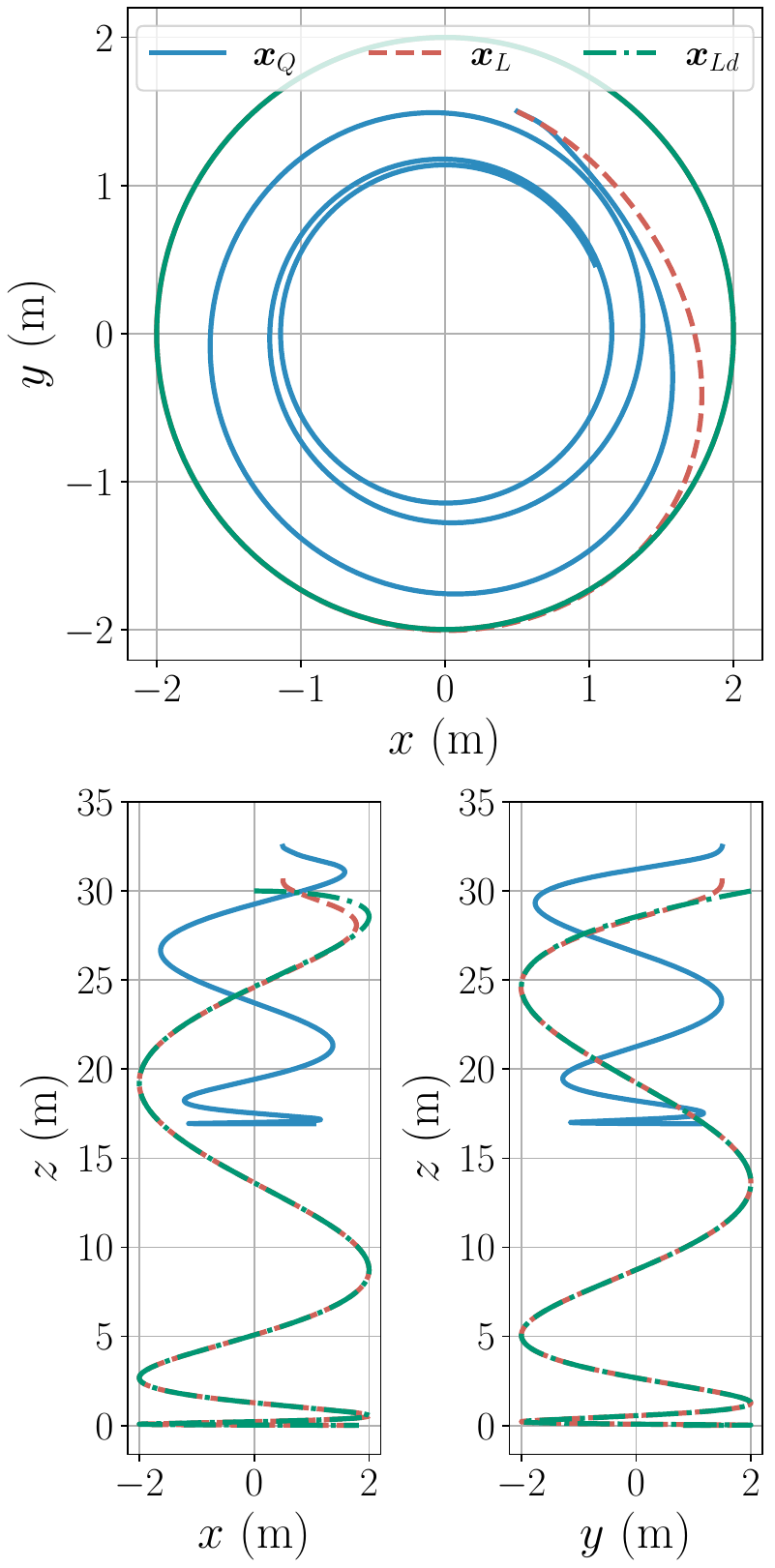}
        \includegraphics[height=2.45in]{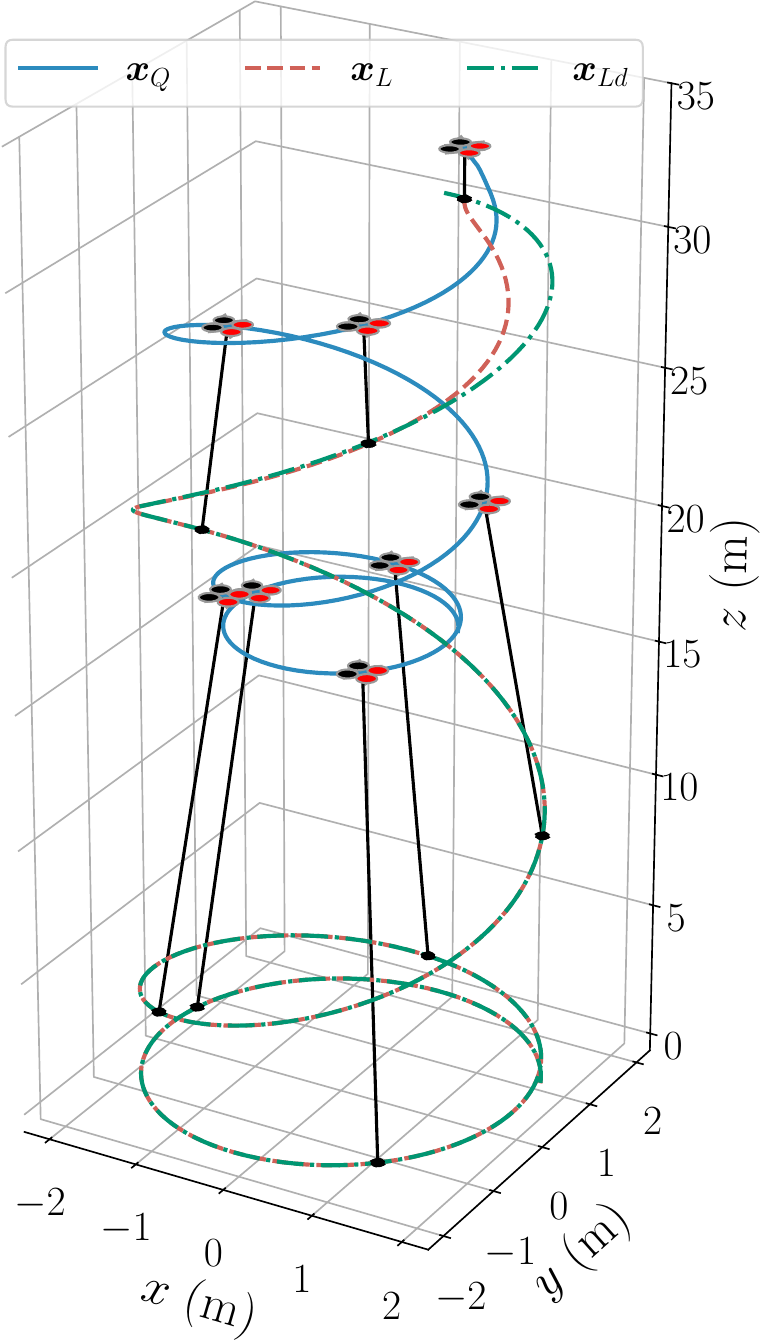}
        \includegraphics[height=2.45in]{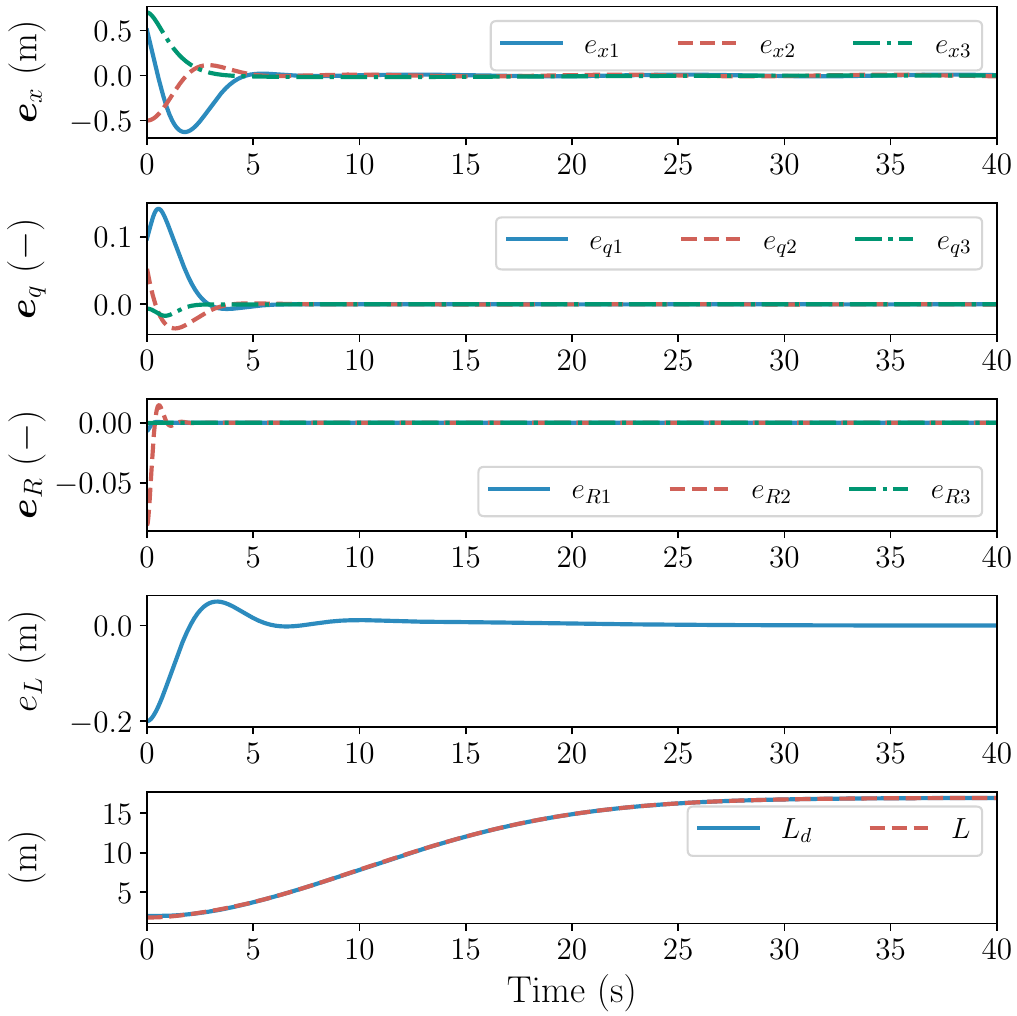}
        \includegraphics[height=2.45in]{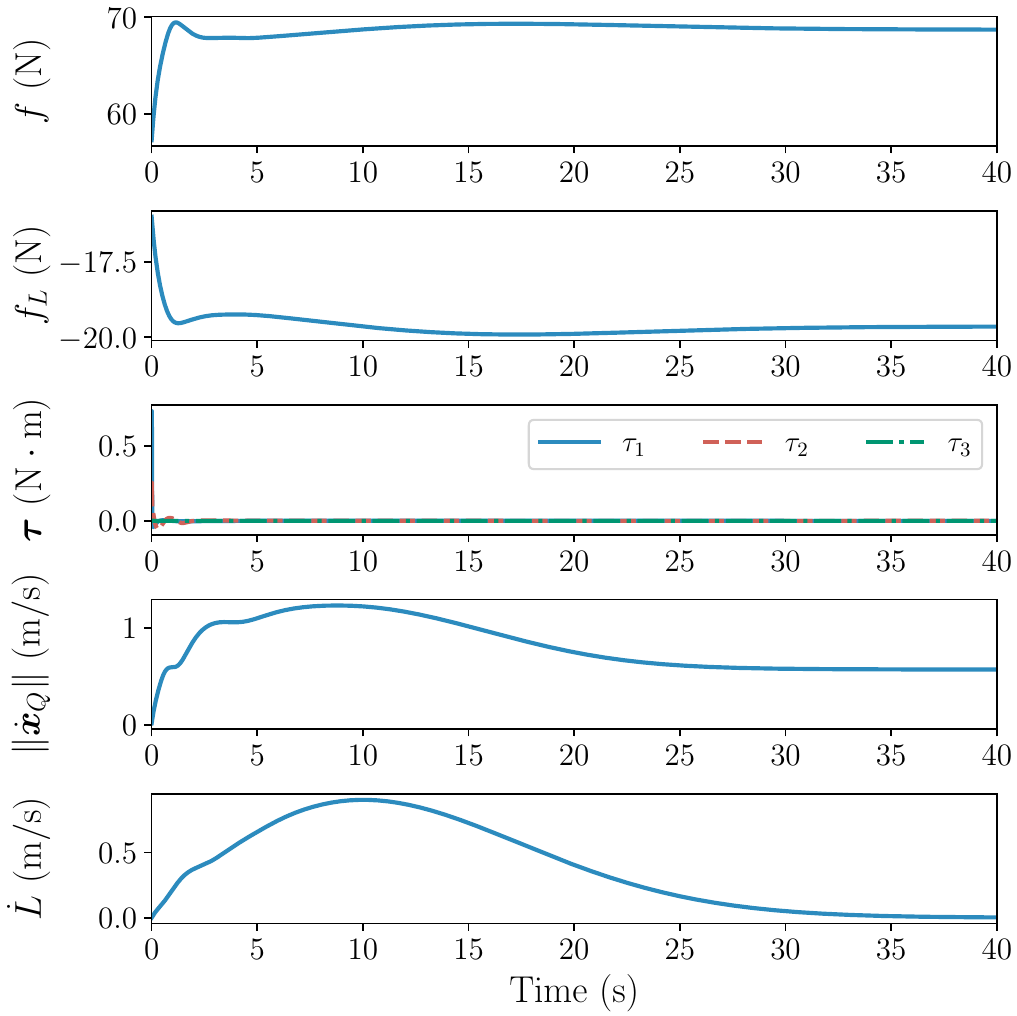}
    }\\
        \hspace*{-10mm} 
    \subfloat[Test3: $k_1 = 100$, $k_2 = 0.1$.]{   
        \includegraphics[height=2.45in]{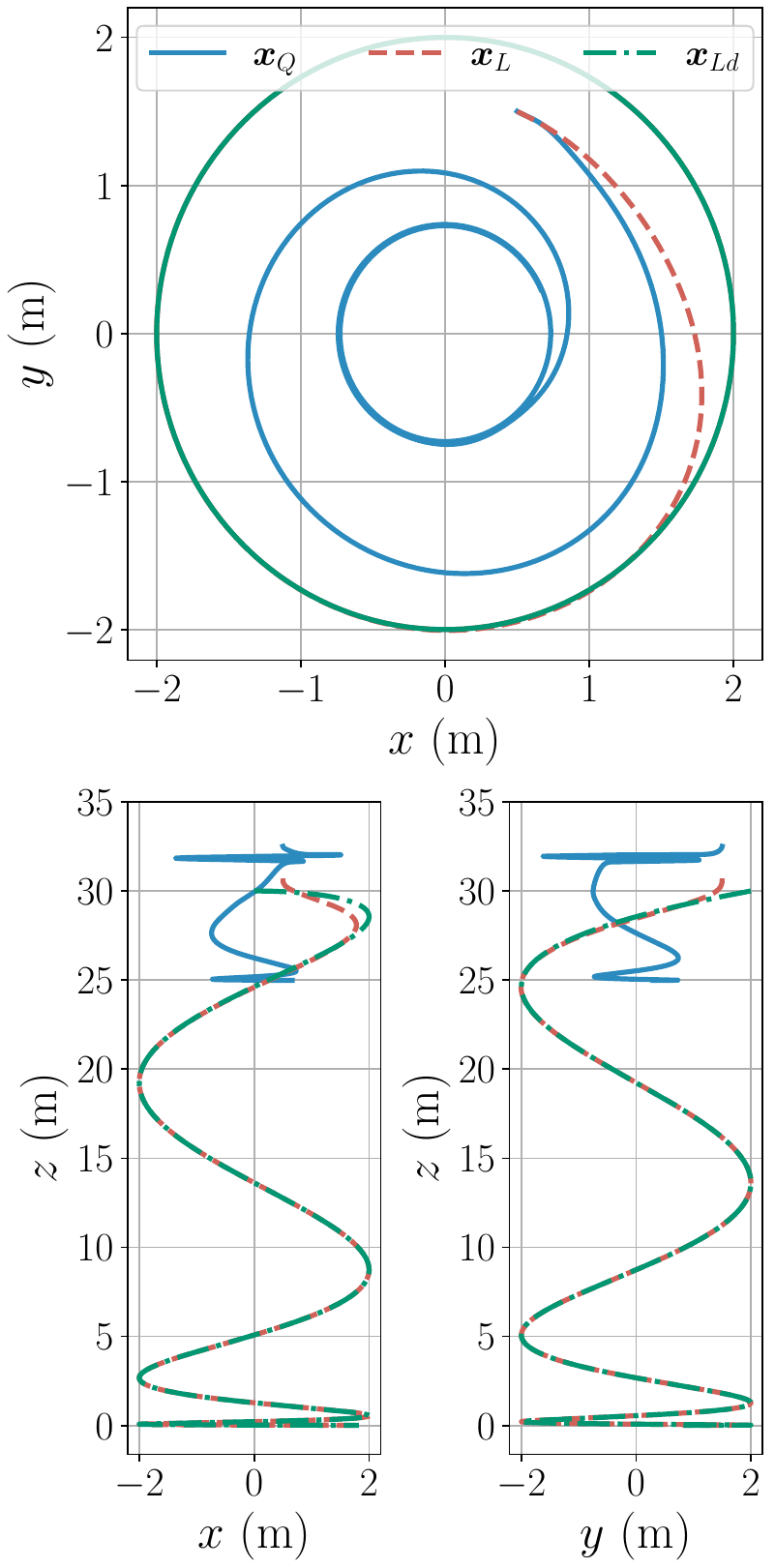}
        \includegraphics[height=2.45in]{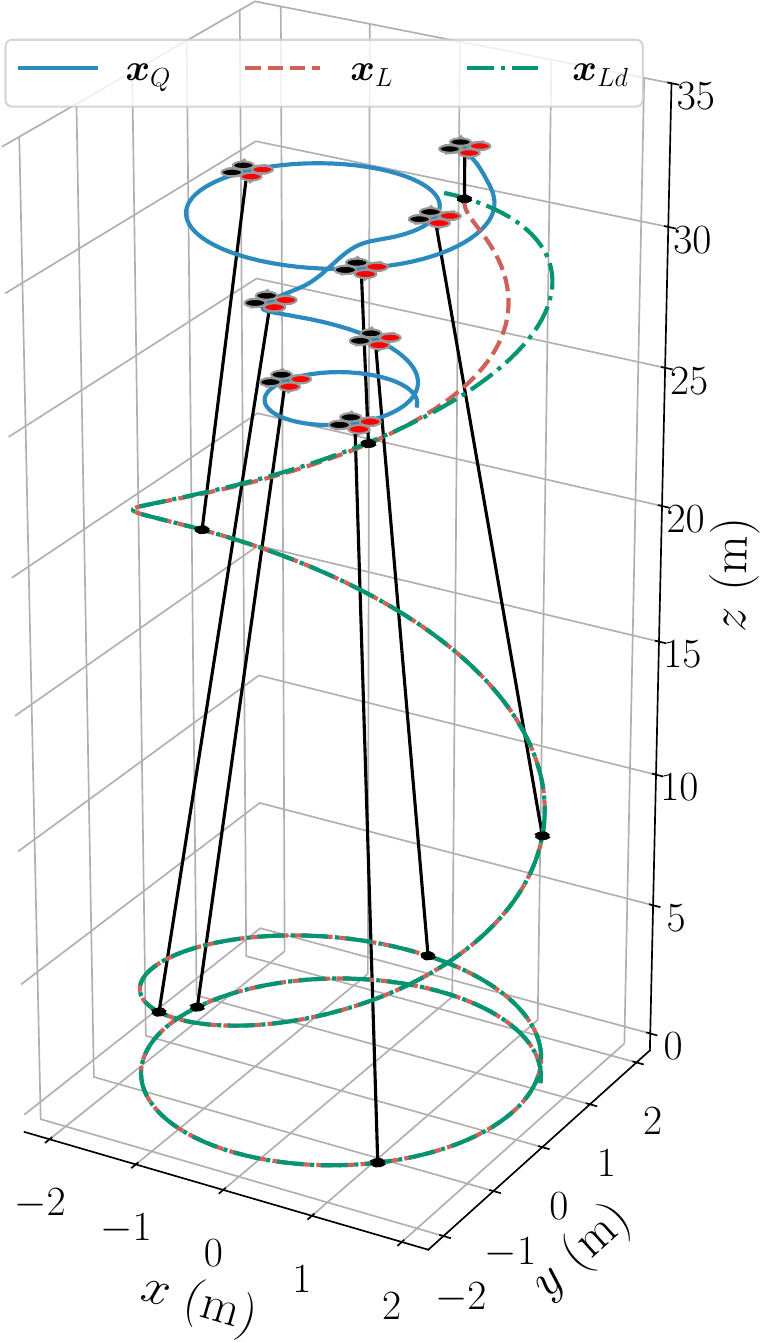}
        \includegraphics[height=2.45in]{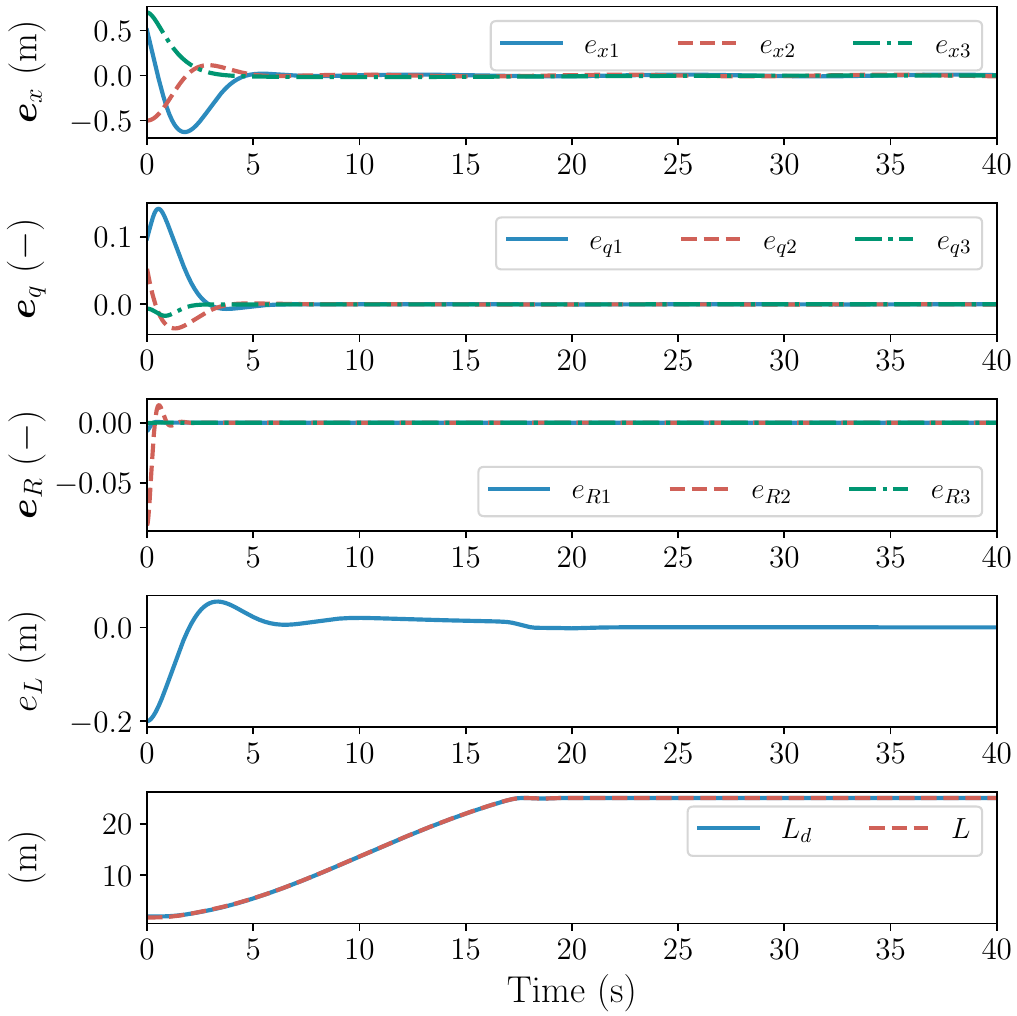}
        \includegraphics[height=2.45in]{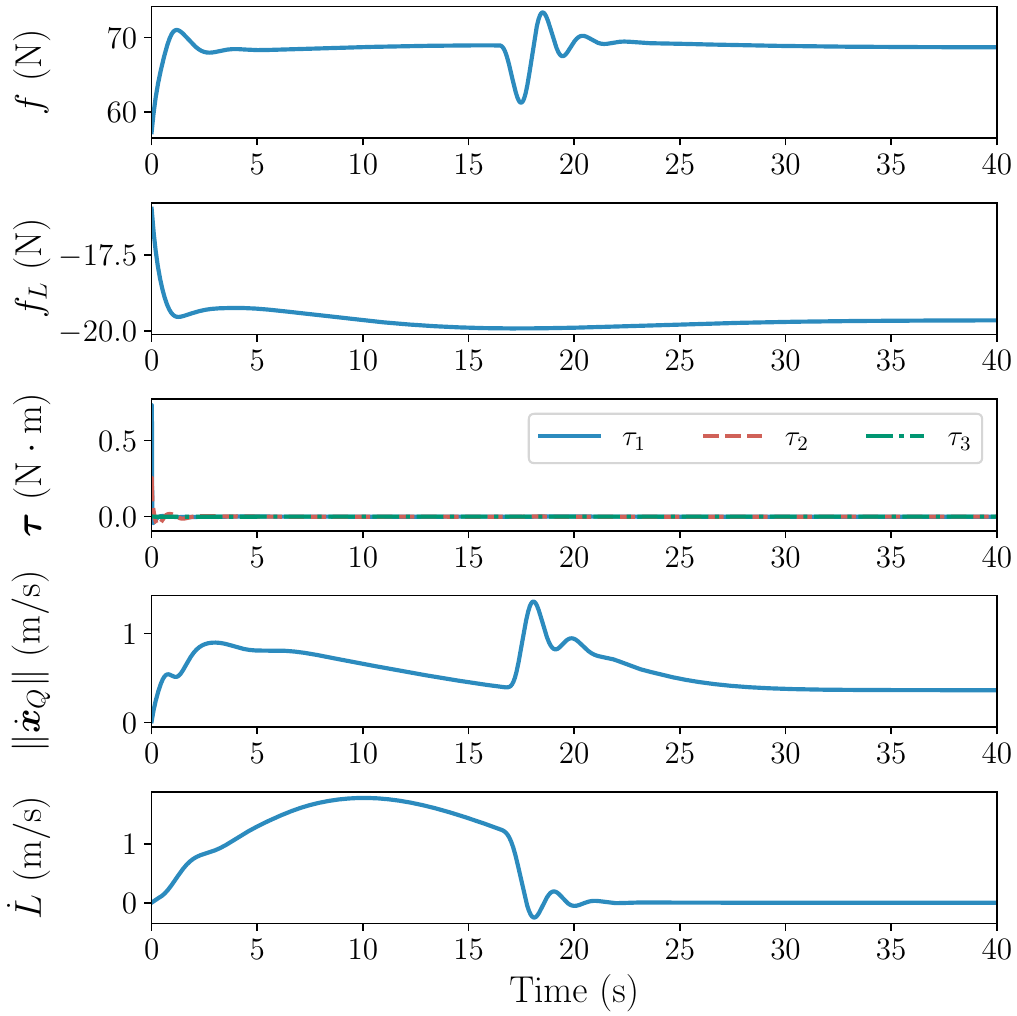}
}
    \caption{\label{fig:Sim_T1} Results for Simulation 1. (The first column presents the two-dimensional trajectories of the payload and the multirotor, while the second column displays their three-dimensional trajectories. The third column depicts the tracking errors for the payload's position, cable direction and length, and multirotor attitude, as well as the generated desired and actual cable length curves. The fourth column shows the multirotor's thrust force and torque, the payload's hoisting/lowering force, and the velocities of both the multirotor and the cable length.)}
\end{figure*}
\section{Simulation Study}\label{sec:simulation_results}

To demonstrate the potential of the solution,  several groups of simulation tests are implemented with the \texttt{acados} library \cite{verschueren2022acados} in this section.  \textcolor{blue}{All simulations are run on a desktop computer equipped with an AMD Ryzen 9 9950X CPU and 64 GB of RAM. The control frequency is 100 Hz, corresponding to a controller sampling period of $T_s = 10~\mathrm{ms}$.}
The physical parameters are selected as $m_Q = 5.0~\mathrm{kg}$, $ m_L = 2.0 ~ \mathrm{kg}$,  $ g = 9.81 \mathrm{m/s^2}$,
$J = \diag([0.05,0.05,0.06])\mathrm{kg \cdot m^2}$. 
The control gains are chosen as $ K_p = \diag([4,4,5]) $, 
$K_d = \diag([4.5, 4.5 ,6])$, $ k_{pl} = 4$, $k_{dl} = 4.5$, $k_q = 1.8$, 
$k_\omega = 1.9$, $k_a = 0.5$, $k_b = 0.5$, $\iota = 0.25$, $ \varrho= \sqrt{0.1}$, $\frac{1}{\epsilon^{2}}k_R = 1.92$, $\frac{1}{\epsilon}k_\Omega = 0.3$.

\subsection{Simulation 1: Basic Performance Validation}
The initial position of the payload is set as $\bm x_L(0) = [0.5, 1.5, 32.5]^\top ~\mathrm{m}$, the initial direction of the cable is $ \bm  q(0) = [0,0,-1]^\top$, and the initial cable length is $L(0) = 1.8 ~\mathrm{m}$. The reference payload trajectory is chosen as 
$\bm x_{Ld} = \left[2\sin(0.5t),2\cos(0.5t),30\exp(-0.005t^2)\right]^\top~\mathrm{m}.$
In this subsection, we focus on the basic performance validation,  and the function $\ell$ is selected to manage the motion of the multirotor with the variation of the cable-length, specifically as: 
\begin{align}
    \ell = k_1 \|\dot{\bm x}_Q\|^2 +k_2 \dot{L}^2 +  \bm L_t^\top K \bm L_t + k_8 {L_d^{(5)}}^2,
\end{align}
where $K = \diag([k_3,k_4,k_5,k_6,k_7])$ and  $k_i, i = 1, ..., 8$ are the  weighting factors.
Three tests are conducted by choosing different weighting factors $k_1 $ and $k_2$, i.e., 
\begin{enumerate}
    \item Test 1:  $k_1 = 0.1$, $k_2 = 100$; 
    \item Test 2:  $k_1 = 100$, $k_2 = 100$;
    \item Test 3:  $k_1 = 100$, $k_2 = 0.1$.
\end{enumerate} 
The rest weighting factors are set the same, i.e., $ K = \diag([0, 1.6,3.2,2.4,0.8])$, $k_8 =0.1$. 
The upper and lower bounds of $\bm L_t$ and $L_d^{(5)}$ are set as $\underline{\bm L}_t = [0.5, -30, -30, -30, -30]^\top$, $\overline{\bm L}_t = [25, 30,30,30,30]^\top$, $\underline{L}_d^{(5)} = -50$, $\overline{L}_d^{(5)} = 50$. 

The simulation results are depicted in Fig. \ref{fig:Sim_T1}. 
The first and second columns illustrate the  two-dimensional and  three-dimensional trajectories of the payload and the multirotor, respectively. 
Analysis of the three tests reveals that when the weight of the cable length velocity in the cost function is high, there is minimal change in the cable length, effectively maintaining a fixed length to achieve precise trajectory tracking of the payload.
Conversely, when the weight of the multirotor velocity is increased, the multirotor exhibits reduced motion variability, and the system compensates by extending the cable length to ensure accurate payload trajectory tracking.
Notably, as demonstrated in Test 3, when constraints are triggered, i.e., the cable reaching its maximum length, the payload trajectory tracking is further achieved by reducing the multirotor's altitude. 
The third column presents the tracking errors for payload position, cable direction and length, and multirotor attitude, along with the generated desired and actual cable length curves. 
Despite initial nonzero errors, all error signals converge to zero within approximately  $5~\mathrm{s}$, validating the efficacy of the proposed control method in enabling the aerial transportation system with variable-length cable to accurately follow the desired trajectory.
The fourth column displays the multirotor thrust force and torque, payload hoisting/lowering force, and the velocities of both the multirotor and cable length.  These curves distinctly illustrate how the multirotor's movement and the variation in cable length are influenced by different weighting factors.
\textcolor{blue}{Table~\ref{tab:acados_time_sim1} reports per-step solve-time statistics including mean, maximum, and standard deviation (std) values. 
Across the three tests, the mean solve time is \(4.931\text{-}5.228~\mathrm{ms}\) (\(49.3\%\text{-}52.3\%\) of \(T_s\)), with std \(0.065\text{-}0.487~\mathrm{ms}\) indicating low jitter overall and a larger variability in Test~3. 
The maximum solve times are \(6.104~\mathrm{ms}\) (Test~2), \(7.710~\mathrm{ms}\) (Test~1), and \(7.753~\mathrm{ms}\) (Test~3), corresponding to utilizations of \(61.0\%\), \(77.1\%\), and \(77.5\%\) of the sampling period \(T_s\), respectively, thus leaving timing headroom of \(3.896~\mathrm{ms}\), \(2.290~\mathrm{ms}\), and \(2.247~\mathrm{ms}\).
These results indicate sufficient real-time margin on the stated hardware.}

\begin{table}[!h]
\color{blue} 
\centering
\caption{\texttt{acados} per-step solve time of Simulation 1.}
\label{tab:acados_time_sim1}
\renewcommand{\arraystretch}{0.85}   
\begin{tabular}{l|ccc}
\toprule
  & mean $(\mathrm{ms})$  & $\max$ $(\mathrm{ms})$ & std $(\mathrm{ms})$ \\
\midrule
Test 1 & 5.114 & 7.710 & 0.065  \\
Test 2 & 4.931 & 6.104 & 0.073 \\
Test 3 & 5.228 & 7.753 & 0.487 \\
\bottomrule
\end{tabular}
\end{table}

\begin{figure}[!tp]
    \centering
    \subfloat[The multirotor altitude, cable length, multirotor velocity, and cable length velocity.\label{fig:T1_quad_alt_velocity}]{   
        \includegraphics[width=3.0in]{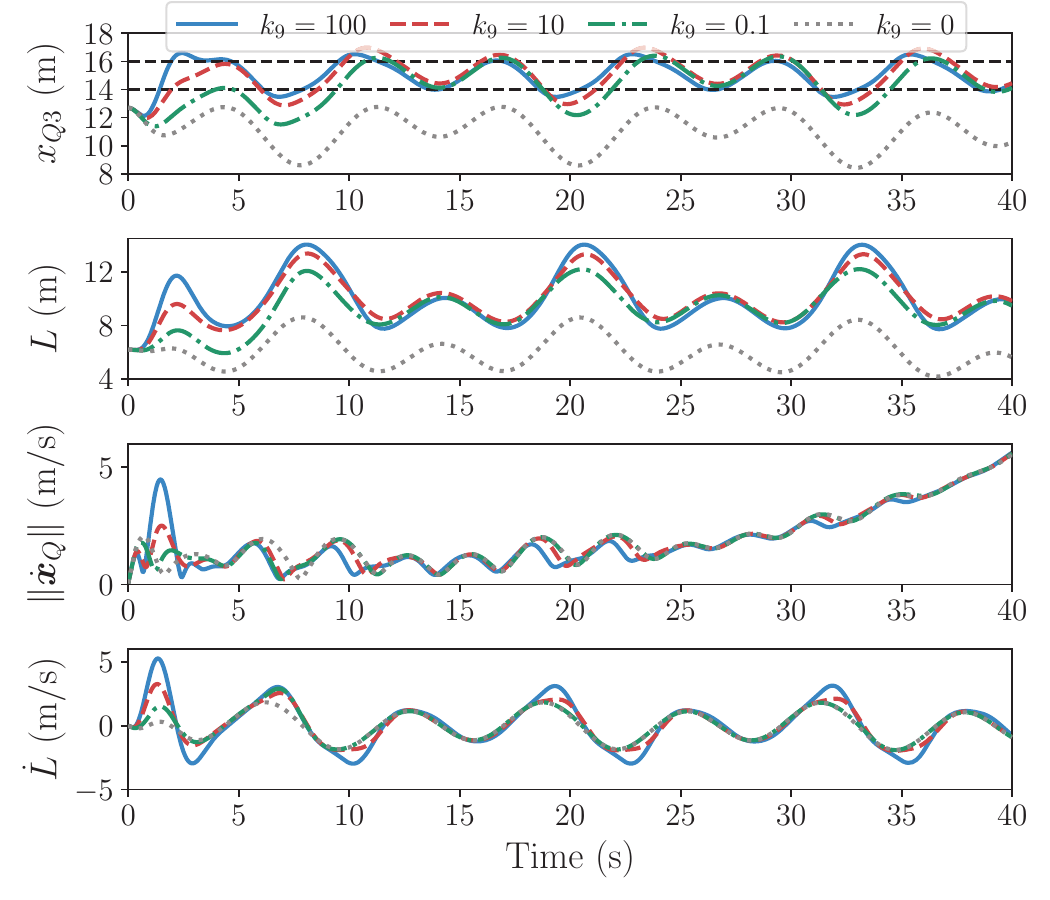}
    }\\ \vspace{-0.2cm}
    \subfloat[The tracking errors of the payload position, cable direction, multirotor attitude, and cable length.\label{fig:T1_quad_alt_error}]{
    \includegraphics[width=3.0in]{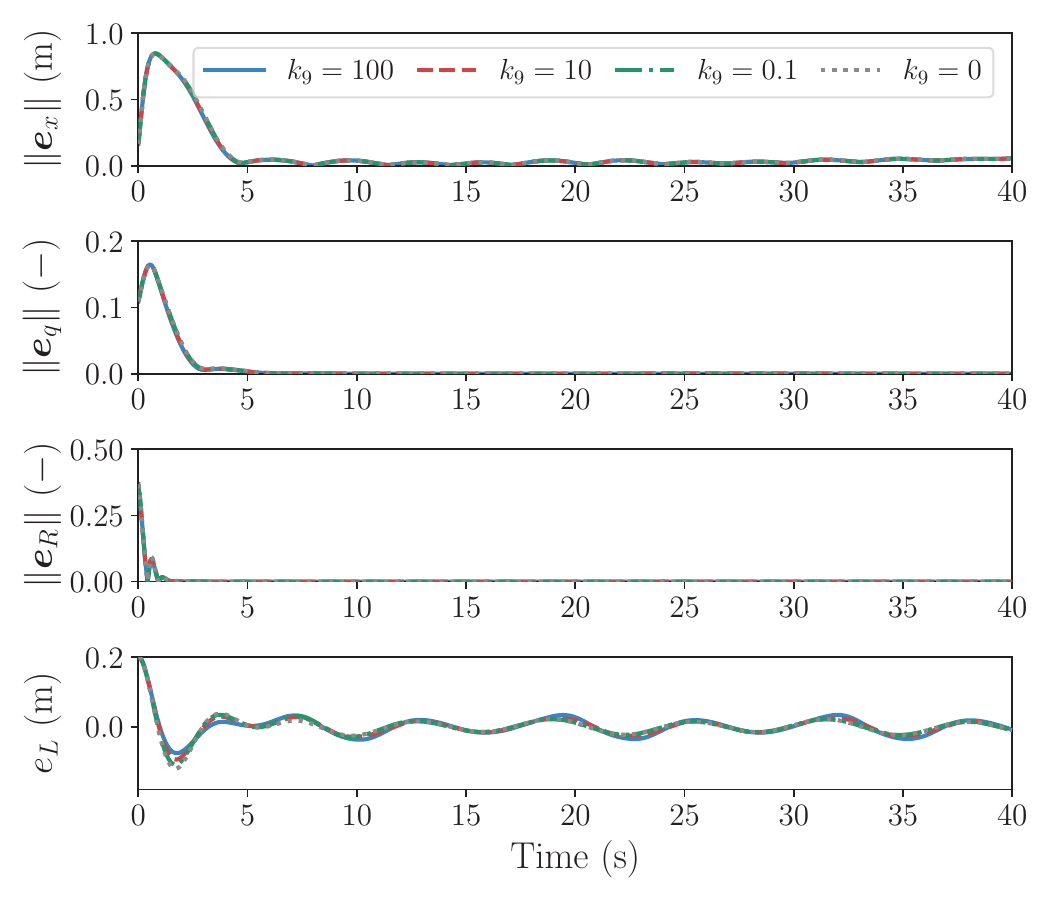}
    }
    \caption{\label{fig:Sim2_T1} Results for Simulation 2 Test 1.}
\end{figure}
\begin{figure}[!tp]
    \centering
    \subfloat[The multirotor altitude, cable length, multirotor velocity, and cable length velocity.\label{fig:T2_quad_alt_velocity}]{   
        \includegraphics[width=3.0in]{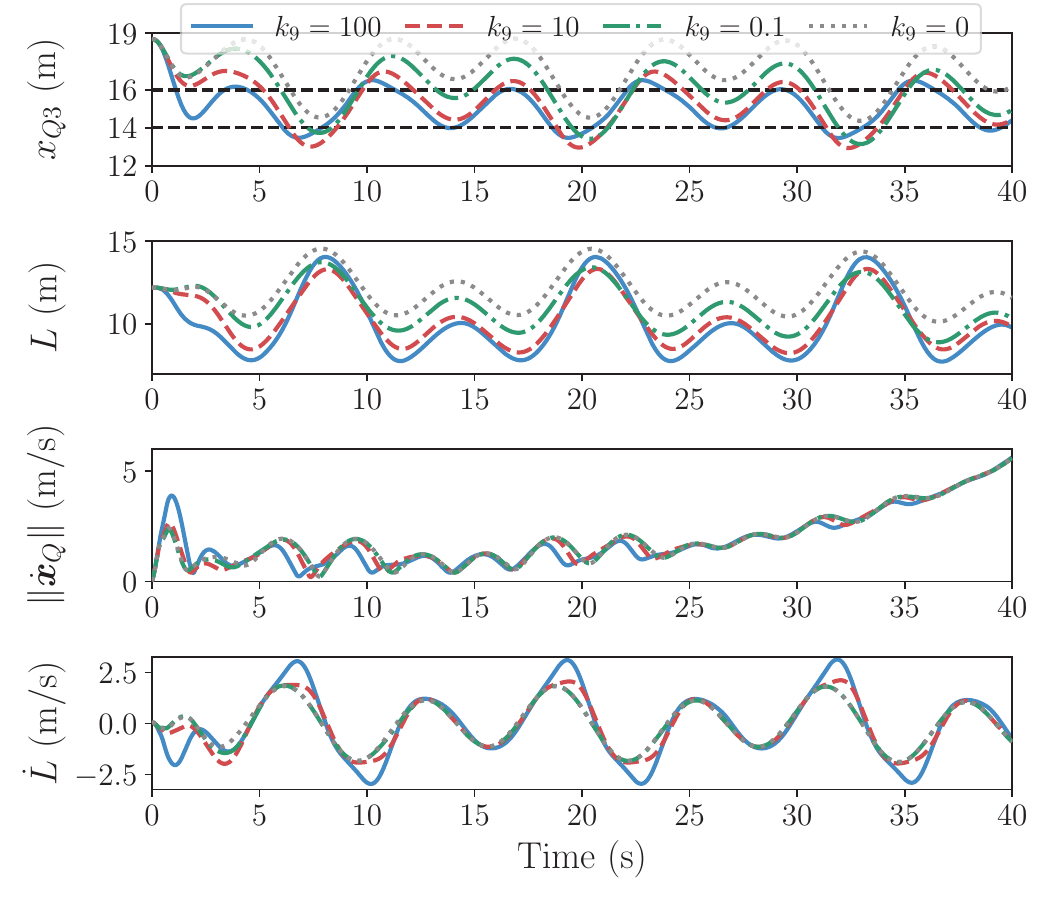}
    }\\ \vspace{-0.2cm}
    \subfloat[The tracking errors of the payload position, cable direction, multirotor attitude, and cable length.\label{fig:T2_quad_alt_error}]{
    \includegraphics[width=3.0in]{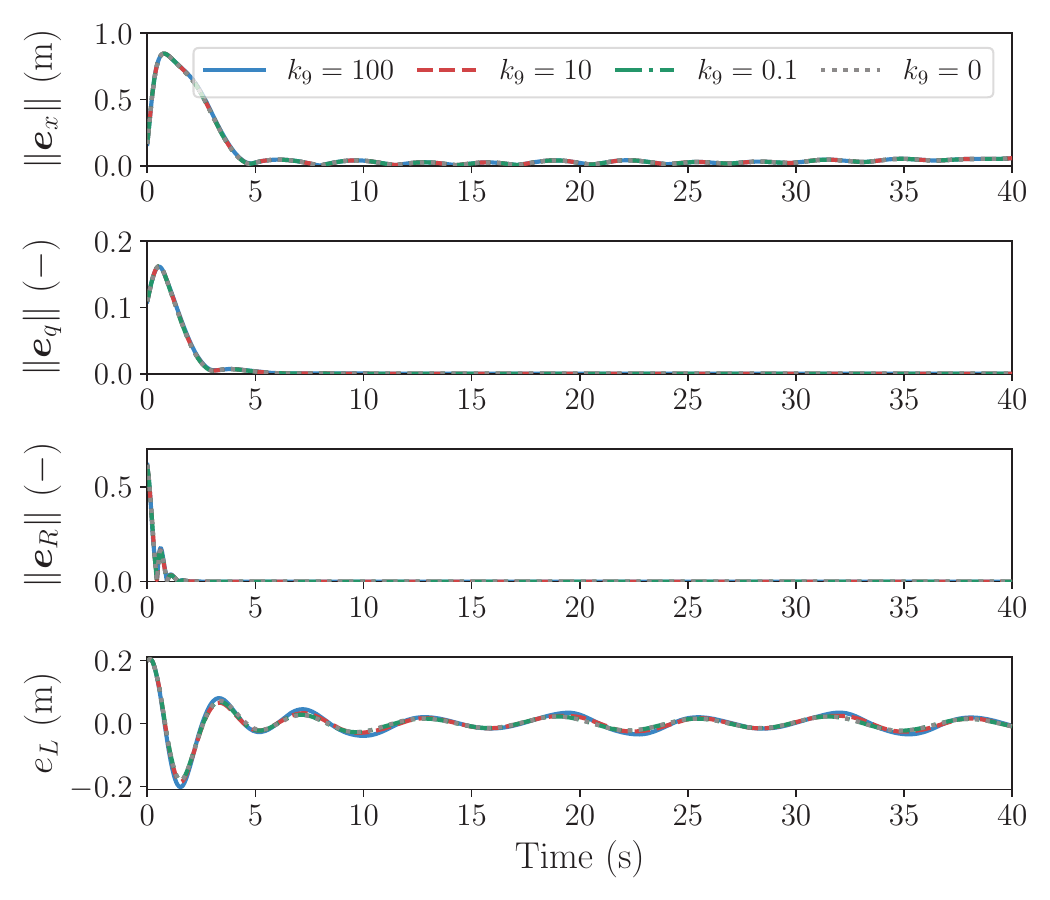}
    }   
    \caption{\label{fig:Sim2_T2} Results for Simulation 2 Test 2.}
\end{figure}
\subsection{Simulation 2: Performance Validation Considering Multirotor Altitude}
During the execution of transportation tasks, it is typically  essential to maintain the multirotor within a specified safe altitude range.  Therefore, to further assess the performance of the cable length generator, a soft constraint on multirotor's altitude is implemented, based on the cost function detailed in the previous subsection, specifically as follows:
\begin{align}
    \ell \!=\! k_1 \|\dot{\bm x}_Q\|^2 +\!k_2 \dot{L}^2 +\!  \bm L_t^\top K \bm L_t + \!k_8 {L_d^{(5)}}^2 +\! k_9 \ell_{\text{alt}},
\end{align}
where $\ell_{\text{alt}} = \exp(k_z(\underline{z_Q}-x_{Q3})) + \exp(k_z(x_{Q3}-\overline{z_Q}))$, $\underline{z_Q}$ and $\overline{z_Q}$ are the expected lower and upper bounds of the multirotor's altitude, and $k_z$ is a positive constant. 
The reference payload trajectory is chosen as 
\begin{align*}
\bm x_{Ld} = \begin{bmatrix}
    \exp(0.1t) \\
    \cos(0.5t) + \sin(0.3t + \frac{\pi}{2}) \\
    2\sin(0.5 t+ \frac{\pi}{4}) +  3\cos( t + \frac{\pi}{2}) +5
    \end{bmatrix}~\mathrm{m}.
\end{align*}
Two sets of tests are conducted by choosing different initial multirotor positions and cable lengths, i.e.,
\begin{enumerate}
    \item Test 1: $x_{Q}(0) = [0.9, 2.1, 12.7]^\top~\mathrm{m}$, $L(0) = 6.2~\mathrm{m}$;
    \item Test 2: $x_{Q}(0) = [0.9, 2.1, 18.7]^\top~\mathrm{m}$, $L(0) = 12.2~\mathrm{m}$.
\end{enumerate}
The target altitude bounds for the multirotor are set between $\underline{z_Q} = 14$ and $\overline{z_Q} = 16$. For both sets of simulations, the initial cable direction is established as $\bm q (0)= [0,0,-1]^\top$. 
Weighting factors are chosen as $k_1 = 100$, $k_2 = 100$, $ K = \diag([0, 1.6,3.2,2.4,0.8])$, $k_8 =0.1$. 
The upper and lower bounds of $\bm L_t$ and $L_d^{(5)}$ are set as $\underline{\bm L}_t = [0.5, -60, -60, -60, -60]^\top$, $\overline{\bm L}_t = [55, 60,60,60,60]^\top$, $\underline{L}_d^{(5)} = -500$, $\overline{L}_d^{(5)} = 500$. 
\begin{table}[!h]
\color{blue} 
\centering
\caption{\texttt{acados} per-step solve time of Simulation 2.}
\label{tab:acados_time_sim2}
\renewcommand{\arraystretch}{1.0}   
\resizebox{0.85\columnwidth}{!}{
\begin{tabular}{lc|ccc}
\toprule
 &Value of $k_9$ & mean (ms) & $\max$ (ms) & std (ms) \\
\midrule
\multirow{4}{*}{Test 1} & 100 & 4.878 & 7.605 & 0.197 \\
                        & 10  & 4.902 & 6.446 & 0.050 \\
                        & 0.1 & 4.922 & 6.187 & 0.059 \\
                        & 0   & 4.895 & 6.472 & 0.135 \\
\hline
\multirow{4}{*}{Test 2} & 100 & 4.682 & 6.513 & 0.093 \\
                        & 10  & 5.126 & 6.652 & 0.152 \\
                        & 0.1 & 4.913 & 5.989 & 0.141 \\
                        & 0   & 4.935 & 5.799 & 0.038 \\
\bottomrule
\end{tabular}}
\end{table}
The simulation results are displayed in Fig. \ref{fig:Sim2_T1} and Fig. \ref{fig:Sim2_T2}.   In each set of simulation, the values of $k_9$ are selected as $100, 10, 0.1$, and $0$ to verify the effectiveness of the cable length generator. 
The results indicate that the term $\ell_{\text{alt}}$ effectively  imposes a constraint on the altitude of the multirotor, with a larger value of $k_9$ leading to a more pronounced constraint on the multirotor's altitude. By comparing Fig. \ref{fig:T1_quad_alt_velocity} and Fig. \ref{fig:T2_quad_alt_velocity}, it can be seen that regardless of whether the multirotor's initial altitude is above or below the specified range,  the incorporation of the cost $\ell_{\text{alt}}$ enables the designed cable length generation scheme to effectively regulate the multirotor's altitude to approach the desired range. Furthermore, Fig. \ref{fig:T1_quad_alt_error} and Fig. \ref{fig:T2_quad_alt_error} demonstrate that under the proposed control scheme, the system errors can quickly converge to zeros. \textcolor{blue}{Table \ref{tab:acados_time_sim2} summarizes the per-step \texttt{acados} solve time in Simulation 2 for four values of $k_9$. 
Across all runs, the mean solve time lies in \(4.682\text{-}5.126\,\mathrm{ms}\) (46.8-51.3\% of \(T_s\)), while the maxima are \(5.799\text{-}7.605\,\mathrm{ms}\) (58.0-76.1\% of \(T_s\)), leaving at least \(2.395\,\mathrm{ms}\) of slack. These timings comfortably meet the 100 Hz real-time requirement on the stated hardware.}

\begin{figure}[t]
    \centering
    \vspace{-0.2cm}
    \subfloat[The tracking errors of the payload position, cable direction, multirotor attitude, and cable length.]{   
        \includegraphics[width=3.2in]{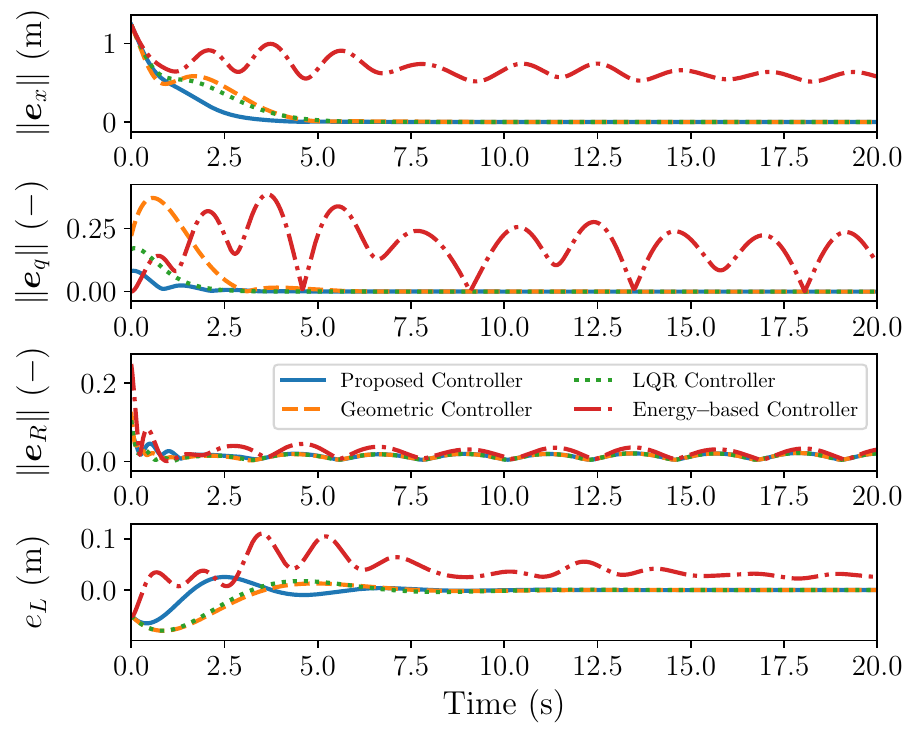}
    }\\ \vspace{-0.2cm}
    \subfloat[The control inputs.]{
    \includegraphics[width=3.2in]{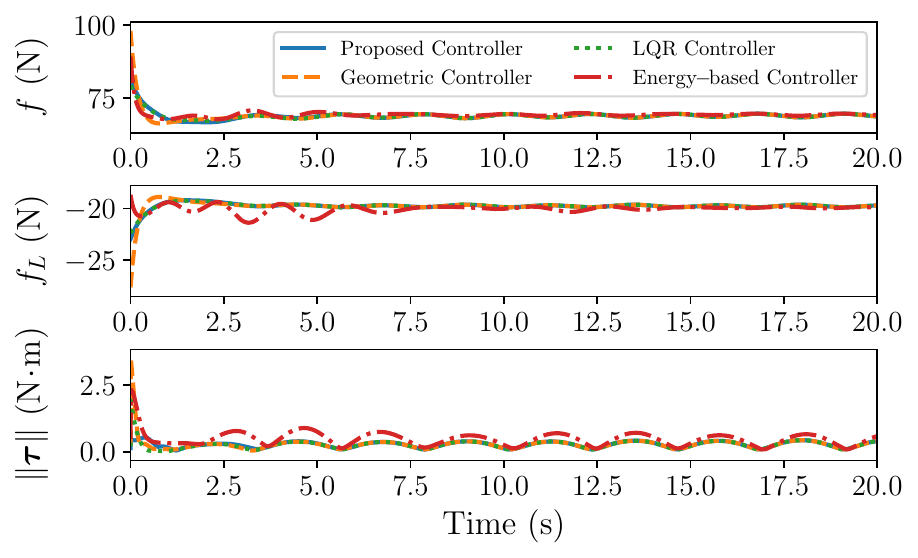}
    }
    \caption{\label{fig:Sim3} Results for Simulation 3.}
\end{figure}
\begin{table*}[t]
\color{blue} 
\centering
\caption{Quantitative results for Simulation 3. RMSE and mean are computed over 20\,s. $t_{\mathrm{conv}}$ is the first time the metric stays within the tolerance for $10$\,s
(0.05 m for $\|\bm e_x\|$, 0.01 m for $|e_L|$, 0.02 for $\|\bm e_q\|$, and 0.02 for $\|\bm e_R\|$).}
\setlength{\tabcolsep}{8pt}
\resizebox{0.99\textwidth}{!}{\renewcommand{\arraystretch}{1.2}
\begin{tabular}{ccccccccccccccc}
\toprule
\multirow{2}{*}{Method} &
\multicolumn{4}{c}{RMSE} &
\multicolumn{4}{c}{mean} &
\multicolumn{4}{c}{$t_{\mathrm{conv}}$ [s]} &
\multirow{2}{*}{Runtime [s]} &
\multirow{2}{*}{\makecell[c]{Average step \\ time [ms]}}\\
\cmidrule(lr){2-5}\cmidrule(lr){6-9}\cmidrule(lr){10-13}
& $\|\bm e_x\|$ [m] & $\|\bm e_q\|$ [--] & $\|\bm e_R\|$ [--] & $e_L$ [m]
& ${\|\bm e_x\|}$ [m] & ${\|\bm e_q\|}$ [--] & ${\|\bm e_R\|}$ [--] & ${|e_L|}$ [m]
& $t_{\|\bm e_x\|}$ & $t_{\|\bm e_q\|}$ & $t_{\|\bm e_R\|}$ & $t_{|e_L|}$
& & \\
\midrule
$\mathrm{Proposed\ Controller}$
& \textbf{0.2030} & \textbf{0.0128} & 0.0164 & \textbf{0.0152}
& \textbf{0.0658} & \textbf{0.0037} & 0.01466 & \textbf{0.0063}
& \textbf{3.12} & \textbf{1.61} & 1.15 & \textbf{3.37}
& 2.595 & 1.297 \\
$\mathrm{Geometric\ Controller}$
& 0.2406 & 0.0934 & 0.0159 & 0.0239
& 0.0992 & 0.0323 & 0.0138 & 0.0106
& 4.34 & 2.79 & 0.62 & 5.78
& 2.515 & 1.257 \\
$\mathrm{LQR\ Controller}$
& 0.2371 & 0.0322 & \textbf{0.0153} & 0.0234
& 0.0952 & 0.0090 & \textbf{0.0135} & 0.0107
& 4.45 & 1.82 & \textbf{0.49} & 5.78
& 2.579 & 1.289 \\
$\mathrm{Energy\!\!-\!\!based\ Controller}$
& 0.6884 & 0.2017 & 0.0311 & 0.0450
& 0.6773 & 0.1838 & 0.0250 & 0.0397
& -- & -- & -- & --
& 1.576 & 0.787 \\
\bottomrule
\end{tabular}}
\label{tab:quantitative_comparison}
\end{table*}
\textcolor{blue}{\subsection{Simulation 3: Comparison with Other Control Methods}
To validate the control performance of the proposed method, a comparative simulation is conducted with the following three methods: the geometric controller \cite{zeng2019geometric}, the LQR controller, and the energy-based controller \cite{yu2023adaptive}. The gains of the controller in \cite{zeng2019geometric} are selected as $k_x = 3.0$, $k_v=3.5$, $ k_L = 4.0$, $k_{\dot{L}} = 4.5$, $k_q = 1.8$, $k_{\dot{q}} = 1.9$. The state weighting matrices of the LQR controller are $Q_x = \mathrm{diag}([12.0, 12.0, 12.0, 8.0, 8.0, 8.0])$, $Q_q = \mathrm{diag}([10.0, 10.0, 10.0, 4.0, 4.0, 4.0])$ and  $Q_L = \mathrm{diag}([10.0, 4.0])$, and the input weighting matrices are $R_x = \mathrm{diag}([1.0, 1.0, 1.0])$, $R_q = \mathrm{diag}([0.5, 0.5, 0.5])$, $R_L = 0.5$. The gains of the controller in \cite{yu2023adaptive} are choosen  $K_p = \mathrm{diag}([5.0, 5.0, 8.0, 10.0])$ and $K_d = \mathrm{diag}([6.0, 6.0, 9.0, 6.0])$.   The reference payload trajectory is selected as a figure-8 curve with a varying cable length, i.e., 
$\bm x_{Ld}= [1.5\sin(0.7t),$ $ 1.5\sin(0.7t)\cos(0.7t), 2.0]^\top~\mathrm{m}$,
$L_d = 1.85 + 0.3\sin(0.25t)~\mathrm{m}$.
To implement the method of \cite{yu2023adaptive}, the desired multirotor trajectory is defined as \(\bm x_{Qd}=\bm x_{Ld}-L_d[0,0,-1]^\top\).
The initial conditions are \(\bm x_L(0)=[0.8,0.5,1.2]^\top~\mathrm{m}\), \(\bm q(0)=[0,0,-1]^\top\), and \(L(0)=1.85~\mathrm{m}\).
The results in Fig.~\ref{fig:Sim3}, Table~\ref{tab:quantitative_comparison} and Table~\ref{table:quantitative_reduction} indicate that the proposed controller yields the lowest tracking errors in payload position, cable direction, and cable length, and achieves shorter convergence times for these quantities.
The attitude RMSE is comparable to that of \cite{zeng2019geometric} and LQR, and lower than \cite{yu2023adaptive}. The average step time is comparable to LQR and \cite{zeng2019geometric}; although \cite{yu2023adaptive} is computationally lighter, it exhibits markedly larger errors and does not meet the predefined convergence criterion.
}
\begin{table}[h]
\color{blue} 
\renewcommand{\arraystretch}{1.0}
\caption{Percentage reductions in RMSE and mean values for $\|\bm e_x\|$,  $\|\bm e_q\|$, $|e_L|$ achieved by the proposed method relative to the comparison methods in Simulation 3.}
\centering
\label{table:quantitative_reduction}
\resizebox{\columnwidth}{!}{
\begin{tabular}{lcccccc}
\toprule
\multirow{2}{*}{Comparison Method} &
\multicolumn{3}{c}{reduction ratio on RMSE} &
\multicolumn{3}{c}{reduction ratio on mean} \\
\cmidrule(lr){2-4}\cmidrule(lr){5-7}
 & $\|\bm e_x\|\,[\mathrm{m}]$ & $\|\bm e_q\|\,[{-}]$  & $e_L\,[\mathrm{m}]$
 & $\|\bm e_x\|\,[\mathrm{m}]$ & $\|\bm e_q\|\,[{-}]$  & $e_L\,[\mathrm{m}]$ \\
\midrule
$\mathrm{Geometric\ Controller}$   & 15.6\% & 86.3\%  & 36.4\%  & 33.7\% & 88.5\% & 40.6\% \\
$\mathrm{LQR\ Controller}$   & 14.4\% & 60.2\% & 35.0\% & 30.9\% & 58.9\% & 41.1\% \\
$\mathrm{Energy\!\!-\!\!based\ Controller}$     & 70.5\% & 93.7\% & 66.2\% & 90.3\% & 98.0\% & 84.1\% \\
\bottomrule
\end{tabular}
}
\end{table}
\begin{figure}[t]
    \centering
    \vspace{-0.2cm}
    \subfloat[The tracking errors of the payload position, cable direction, multirotor attitude, and cable length.]{   
        \includegraphics[width=3.2in]{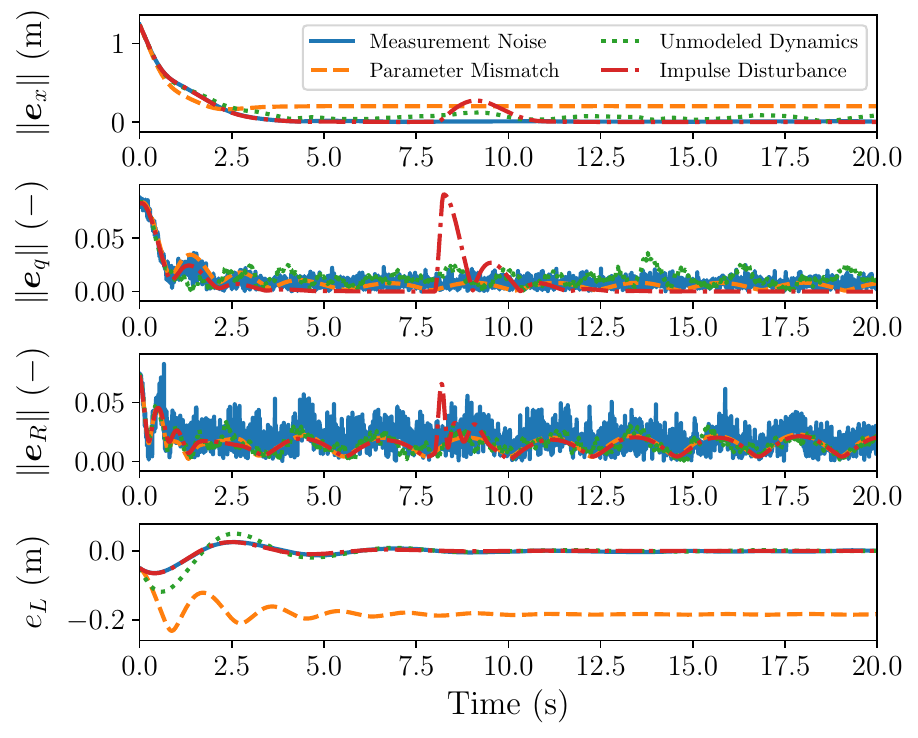}
    }\\ \vspace{-0.2cm}
    \subfloat[The control inputs.]{
    \includegraphics[width=3.2in]{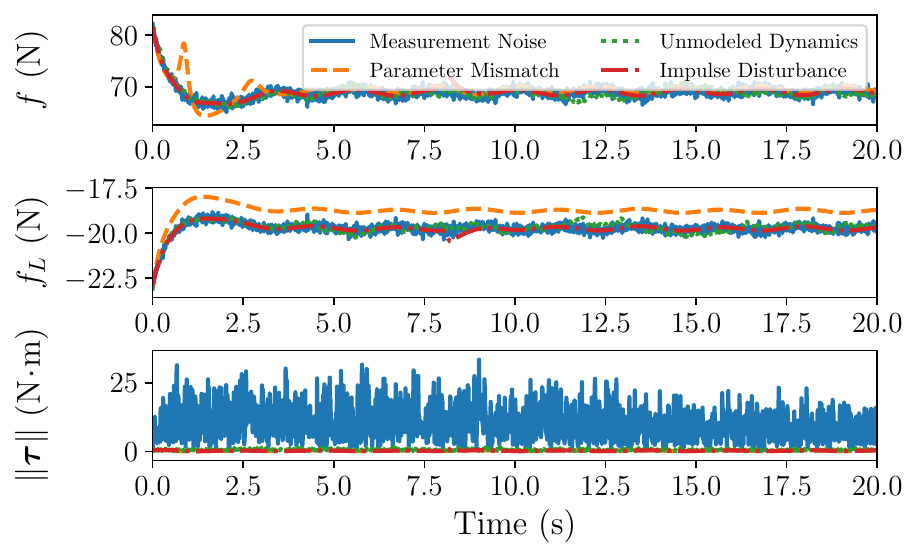}
    }
    \caption{\label{fig:Sim4} Results for Simulation 4.}
\end{figure}
\textcolor{blue}{\subsection{Simulation 4: Robustness Validation}\label{sec:sim4}
To test the robustness of the proposed method, four cases of uncertainty/disturbance are injected to the system: (i) \emph{Measurement noise:} zero-mean Gaussian noise on payload position, velocity, cable length, and cable-length variation rate with standard deviation (std.) \(0.02\) (units: m, m/s, m, m/s, respectively); multirotor angular velocity noise with std.\ \(0.01~\mathrm{rad/s}\) per axis; and a cable-direction perturbation given by a \(0.01~\mathrm{rad}\) axis--angle rotation that preserves \(\|\bm q\|=1\). (ii) \emph{Parameter mismatch:} \(m_Q=1.03\,m_{Q\mathrm{nom}}\), \(m_L=0.95\,m_{L\mathrm{nom}}\), and \(J=\mathrm{diag}(1.10,\,0.95,\,1.05)\,J_{\mathrm{nom}}\). (iii) \emph{Unmodeled dynamics:} an actuator lag of \(50~\mathrm{ms}\) together with a stochastic wind acceleration acting on the payload as zero-mean Gaussian noise with std.\ \([0.6,\,0.6,\,0.6]^\top~\mathrm{m/s^2}\).
(iv) \emph{Impulse disturbance:} a half-sine force in the world frame at \(t=8~\mathrm{s}\) lasting \(0.25~\mathrm{s}\) with amplitude \([5,\,1,\,-2]^\top~\mathrm{N}\). The reference trajectory and initial state are the same as in Simulation 3. Fig. \ref{fig:Sim4} presents the result curves, showing that the closed-loop system remains stable with bounded tracking errors across all four scenarios, demonstrating the satisfactory robustness of the proposed control scheme.
}
\vspace{-0.5cm}

\textcolor{blue}{\subsection{Discussion of Practical Implementation}
Under the proposed control framework, once a desired payload trajectory is specified, the cable-length generator produces a $C^5$-continuous cable-length trajectory. The commanded quantities: the cable direction $\bm q_d$, cable angular velocity $\bm \omega_d$, rotation matrix $R_d$, multirotor angular velocity $\bm \Omega_d$, and multirotor angular acceleration $\dot{\bm \Omega}_d$, are computed from the closed-loop system dynamics via systematic derivations. Consequently, the controllers in \eqref{Control:Fq}, \eqref{Control:fc}, \eqref{Control:Fcperp}, and \eqref{Control:tau} are algebraic state-feedback laws that rely only on directly measured signals; they require neither acceleration inputs nor higher-order state derivatives, thereby avoiding numerical differentiation of noisy measurements.
For sensing, the rotation matrix $R$ and angular velocity $\bm \Omega$ can be measured by the onboard IMU. Indoors, the multirotor/payload positions $\bm x_Q, \bm x_L$ and velocities $\bm v_Q, \bm v_L$ can be obtained from a motion-capture system, from which the cable length $L$, its rate $\dot L$, the cable direction $\bm q$, and its angular velocity $\bm \omega$ are computed geometrically. Outdoors, $\bm x_Q$ and $\bm v_Q$ come from GPS or visual inertial odometry, etc.; a winch-mounted encoder measures $L$ and $\dot L$; and an IMU at the cable end provides $\bm q$ and $\bm \omega$. The payload states $\bm x_L,\bm v_L$ are then inferred from the known multirotor-payload geometry. 
Unlike fixed-length systems, an electric winch driven by a current-controlled motor can be used to adjust the cable length. Given the torque-current relation $\tau_m = k_m I_m$ and the drum radius $r_s$, the cable tension satisfies $f_L = \frac{\tau_m}{r_s}$, implying $I_m = \frac{r_s}{k_m}f_L$, where $k_m$ denotes the motor torque constant. 
Accordingly, the controller commands the winch-motor current $I_m$ to achieve the commanded cable tension during hoisting and lowering. 
Finally, the vehicle-level inputs $(f,\bm\tau)$ are mapped to rotor speeds via the mixer, and the hoisting/lowering command $f_L$ is sent to the winch motor \cite{huang2023suppressing, liang2022unmanned, yu2023adaptive}.}
\vspace{-0.2cm}

\section{Conclusions}\label{sec:conclusions}
This paper explores a payload trajectory tracking scheme for the aerial transportation system with variable-length cable. The proposed scheme includes a backstepping controller for payload trajectory, cable length and direction tracking, as well as an online cable length generator. The generator integrates the backstepping controller into the dynamic system model constraints and optimizes the desired cable length in real-time for the control scheme to follow. By employing Lyapunov techniques and growth restriction conditions, the proposed control scheme ensures the asymptotic stability of the closed-loop system. Under a reasonable cost function, the cable length generator effectively allocates the multirotor's motion and cable length, facilitating consistent payload trajectory tracking across various configurations of multirotor positions and cable lengths. Simulation results further highlight the advantages of the proposed control scheme, demonstrating its effectiveness in both trajectory tracking and cable length generation. 
\textcolor{blue}{In future work, we will investigate novel control schemes to strengthen the theoretical robustness guarantees of variable-length cable-suspended aerial transportation systems.}

\appendix
\renewcommand{\appendixname}{Appendix}
\section{Calculation of $\dot{\boldsymbol{e}}_q^\top \boldsymbol{e}_\omega $}\label{App:A}
To calculate term $\dot{\bm e}_q^\top \bm e_\omega$, by using the fact $\bm \omega^\top \bm q= 0 $, $ \bm \omega_d^\top \bm q_d =0 $, one first derives the derivation of $\dot{\bm e}_q$ as follows:
\begin{align*}
    \dot{\bm e}_q 
    =&- \bm q\times(\bm \omega_d \times \bm q_d) + \bm q_d\times(\bm \omega\times\bm q)\nonumber\\
    =&-\bm \omega_d (\bm q^\top \bm q_d) + \bm q_d (\bm q^\top \bm \omega_d) + \bm \omega(\bm q_d^\top\bm q) -\bm q(\bm q_d^\top\bm\omega)\nonumber\\
    =&(\bm \omega +\hat{\bm q}^2\omega_d - \hat{\bm q}^2\omega_d- \bm \omega_d)(\bm q^\top\bm q_d)\nonumber\\
    & + \bm q_d (\bm q^\top \bm \omega_d) -\bm q (\bm \omega_d^\top \bm q_d) + \bm q_d (\bm \omega^\top \bm q)-\bm q(\bm q_d^\top\bm\omega)\nonumber\\
    =&\bm e_\omega(\bm q^\top\bm q_d) -(\bm q^\top \bm \omega_d)(\bm q^\top\bm q_d)\bm q\nonumber\\
    &+ \bm \omega_d \times (\bm q_d \times \bm q) + \bm \omega\times(\bm q_d\times \bm q)\nonumber\\
    =&\bm  e_\omega(\bm q^\top\bm q_d) -(\bm q^\top \bm \omega_d)(\bm q^\top\bm q_d)\bm q + \bm e_\omega \times \bm e_q \nonumber\\
    &+ (\bm \omega_d -\hat{\bm q}^2\bm \omega_d)\times \bm e_q.
\end{align*}
Subsequently, the derivation of $\dot{\bm e}_q^\top \bm e_\omega$ can be obtained as follows:
\begin{align*}
    \dot{\bm e}_q^\top \bm e_\omega 
 =&(\bm q^\top\bm q_d) \bm e_\omega^\top \bm e_\omega+(2\bm \omega_d -\bm \omega_d -\hat{\bm q}^2\bm \omega_d)\times \bm e_q^\top\bm e_\omega\nonumber\\
 =&(\bm q^\top\bm q_d) \bm e_\omega^\top \bm e_\omega+(2\bm \omega_d -(\bm q^\top\bm \omega_d)\bm q)\times \bm e_q^\top\bm e_\omega\nonumber\\
=&(\bm q^\top\bm q_d) \bm e_\omega^\top \bm e_\omega+((2I_{3\times3}-\bm q \bm q^\top)\bm \omega_d)\times \bm e_q^\top\bm e_\omega\nonumber\\
\leq &\|\bm e_\omega\|^2 + C_\omega\|\bm e_q\|\|\bm e_\omega\|,
\end{align*}
where $C_\omega = \sup \left(\left\|(2I_{3\times3}-\bm q \bm q^\top)\bm \omega_d\right\|\right)$.

\bibliographystyle{plain}        
\bibliography{autosam}           

\end{document}